\def\year{2020}\relax
\documentclass[letterpaper]{article} 
\usepackage{layout}

\usepackage{amsfonts}
\usepackage{amsmath}
\usepackage{amsthm}
\usepackage{amssymb} 

\usepackage{mdframed} 
\usepackage{thmtools}

\definecolor{shadecolor}{gray}{0.95}
\declaretheoremstyle[
headfont=\normalfont\bfseries,
notefont=\mdseries, notebraces={(}{)},
bodyfont=\normalfont,
postheadspace=0.5em,
spaceabove=1pt,
mdframed={
  skipabove=8pt,
  skipbelow=8pt,
  hidealllines=true,
  backgroundcolor={shadecolor},
  innerleftmargin=4pt,
  innerrightmargin=4pt}
]{shaded}

 \usepackage{xcolor}
\usepackage{color}
\usepackage{graphicx}

\usepackage{verbatim}

\newcommand{\R}{\mathbb{R}} 

\newcommand{\cC}{{\cal C}}
\newcommand{\cD}{{\cal D}}

\newcommand{\cL}{{\cal L}}

\newcommand{\cO}{{\cal O}}

\newcommand{\mA}{{\bf A}}
\newcommand{\mB}{{\bf B}}

\newcommand{\mI}{{\bf I}}

\newcommand{\mM}{{\bf M}}

\newcommand{\mW}{{\bf W}}
\newcommand{\mX}{{\bf X}}

\usepackage[colorinlistoftodos,bordercolor=orange,backgroundcolor=orange!20,linecolor=orange,textsize=scriptsize]{todonotes}

\newcommand{\eqdef}{:=}

\newtheorem{assumption}{Assumption}
\newtheorem{lemma}{Lemma}

\newtheorem{theorem}{Theorem}
\newtheorem{proposition}{Proposition}

\theoremstyle{definition}
\newtheorem{definition}[theorem]{Definition}
\theoremstyle{remark}
\newtheorem{remark}{Remark} 
\usepackage{bbm}

\usepackage{aaai20}  
\usepackage{times}  
\usepackage{helvet} 
\usepackage{courier}  
\usepackage{url}  
\usepackage{graphicx} 
\usepackage{graphicx}  
\usepackage{listings}
\usepackage{enumerate}

\usepackage[colorinlistoftodos,bordercolor=orange,backgroundcolor=orange!20,linecolor=orange,textsize=scriptsize]{todonotes}

\newcommand{\smartparagraph}[1]{\noindent {\bf #1}}

\usepackage{algorithm}
\usepackage{algorithmic}
\usepackage[utf8]{inputenc} 
\usepackage{booktabs}       
\usepackage{amsfonts}       
\usepackage{nicefrac}       
\usepackage{microtype}      

\usepackage[T1]{fontenc}
\usepackage{mwe}    
\usepackage{subcaption}
\usepackage{lineno}
\usepackage{xcolor}
 \newtheorem*{theorem*}{Theorem}
 \newtheorem*{proposition*}{Proposition}
 \newtheorem*{lemma*}{Lemma}

\definecolor{codegreen}{rgb}{0,0.6,0}
\definecolor{codegray}{rgb}{0.5,0.5,0.5}
\definecolor{codepurple}{rgb}{0.58,0,0.82}
\definecolor{backcolour}{rgb}{0.95,0.95,0.92}

\lstdefinestyle{mystyle}{
    backgroundcolor=\color{backcolour},   
    commentstyle=\color{codegreen},
    keywordstyle=\color{magenta},
    numberstyle=\tiny\color{codegray},
    stringstyle=\color{codepurple},
    basicstyle=\ttfamily\footnotesize,
    breakatwhitespace=false,         
    breaklines=true,                 
    captionpos=b,                    
    keepspaces=true,                 
    numbers=left,                    
    numbersep=5pt,                  
    showspaces=false,                
    showstringspaces=false,
    showtabs=false,                  
    tabsize=2
}
 
\author{
Aritra Dutta, El Houcine Bergou, Ahmed M. Abdelmoniem, Chen-Yu Ho, Atal Narayan Sahu, Marco Canini, Panos Kalnis)
} 

\title{On the Discrepancy between the Theoretical Analysis and Practical Implementations of Compressed Communication for Distributed Deep Learning}
\author{
Aritra Dutta, El Houcine Bergou\affmark[1], Ahmed M. Abdelmoniem,\\
\Large \textbf{Chen-Yu Ho, Atal Narayan Sahu, Marco Canini, Panos Kalnis}\\
KAUST\\INRA\affmark[1]
}

\author{
Aritra Dutta, El Houcine Bergou\thanks{The author is also with INRA}, Ahmed M. Abdelmoniem,\\
\Large \textbf{Chen-Yu Ho, Atal Narayan Sahu, Marco Canini, Panos Kalnis}\\
KAUST}

\newcommand*{\affmark}[1][*]{\textsuperscript{#1}}

\begin{document}
\maketitle

\begin{abstract}

Compressed communication, in the form of sparsification or quantization of stochastic gradients, is employed to reduce communication costs in distributed data-parallel training of deep neural networks. However, there exists a discrepancy between theory and practice: while theoretical analysis of most existing compression methods assumes compression is applied to the gradients of the entire model, many practical implementations operate individually on the gradients of each layer of the model.

In this paper, we prove that layer-wise compression is, in theory, better, because the convergence rate is upper bounded by that of entire-model compression for a wide range of biased and unbiased compression methods. However, despite the theoretical bound, our experimental study of six well-known methods shows that convergence, in practice, may or may not be better, depending on the actual trained model and compression ratio. Our findings suggest that it would be advantageous for deep learning frameworks to include support for both layer-wise and entire-model compression. 
\end{abstract}

\section{Introduction}

Despite the recent advances in deep learning and its wide-spread transformative successes, training deep neural networks (DNNs) remains a computationally-intensive and time-consuming task. The continuous trends towards larger volumes of data and bigger DNN model sizes require to scale out training by parallelizing the optimization algorithm across a set of workers.
The most common scale out strategy is data parallelism where each worker acts on a partition of input data. In each iteration of the optimization algorithm -- typically the stochastic gradient descent (SGD) algorithm -- every worker processes a
mini-batch of the input data and produces corresponding stochastic gradients.
Then, gradients from all workers are aggregated to produce an update to model parameters to be applied prior to the next iteration.
The gradient aggregation process involves network communication and is supported via a parameter-server architecture or collective communication routines (e.g., {\tt all\_reduce}).

A major drawback of distributed training across parallel workers is that training time can negatively suffer from high communication costs, especially at large scale, due to the transmission of stochastic gradients. To alleviate this problem, several lossy compression techniques have been proposed \cite{DBLP:conf/interspeech/SeideFDLY14,Dettmers15,alistarh2017qsgd,terngrad,Alistarh18_sparse,signsgd,Tang_2019,WangniWLZ18,cnat}. 

Two main classes of compression approaches are sparsification and quantization.
{\em Sparsification} communicates only a subset of gradient elements.
For instance, this is obtained by selecting uniformly at random $k\%$ elements or the top $k\%$ elements by magnitude~\cite{Alistarh18_sparse}.
{\em Quantization}, on the other hand, represents gradient elements with lower precision, thus using fewer bits for each element. For instance, this is done by transmitting only the sign of each element~\cite{signsgd}, or by randomized rounding to a discrete set of values~\cite{alistarh2017qsgd}.

In theory, such methods can reduce the amount of communication and their analysis reveal that they provide convergence guarantees (under certain analytic assumptions). Further, such methods preserve model accuracy across a range of settings in practice.

However, we observe a discrepancy between the theoretical analysis and practical implementation of existing compression methods.
To the best of our knowledge, the theoretical analysis of every prior method appears to assume that compression is applied to the gradient values of the \emph{entire model}. However, from our study of existing implementations \cite{poseidon,Shi2017,Lim2019,Shi2019,cnat} and experience with implementing compression methods, we observe that compression is applied {\em layer by layer}, as illustrated in Figure~\ref{fig:layerwise_architec}.
In fact, based on the existing programming interfaces in modern distributed machine learning toolkits such as PyTorch \cite{pytorch} and TensorFlow \cite{tensorflow}, a layer-wise implementation is typically most straightforward because wait-free backpropagation~\cite{poseidon} -- where gradients are sent as soon as they are available -- is a commonly used optimization.

\begin{figure}[t]
    \centering
    \includegraphics[width=0.4\textwidth]{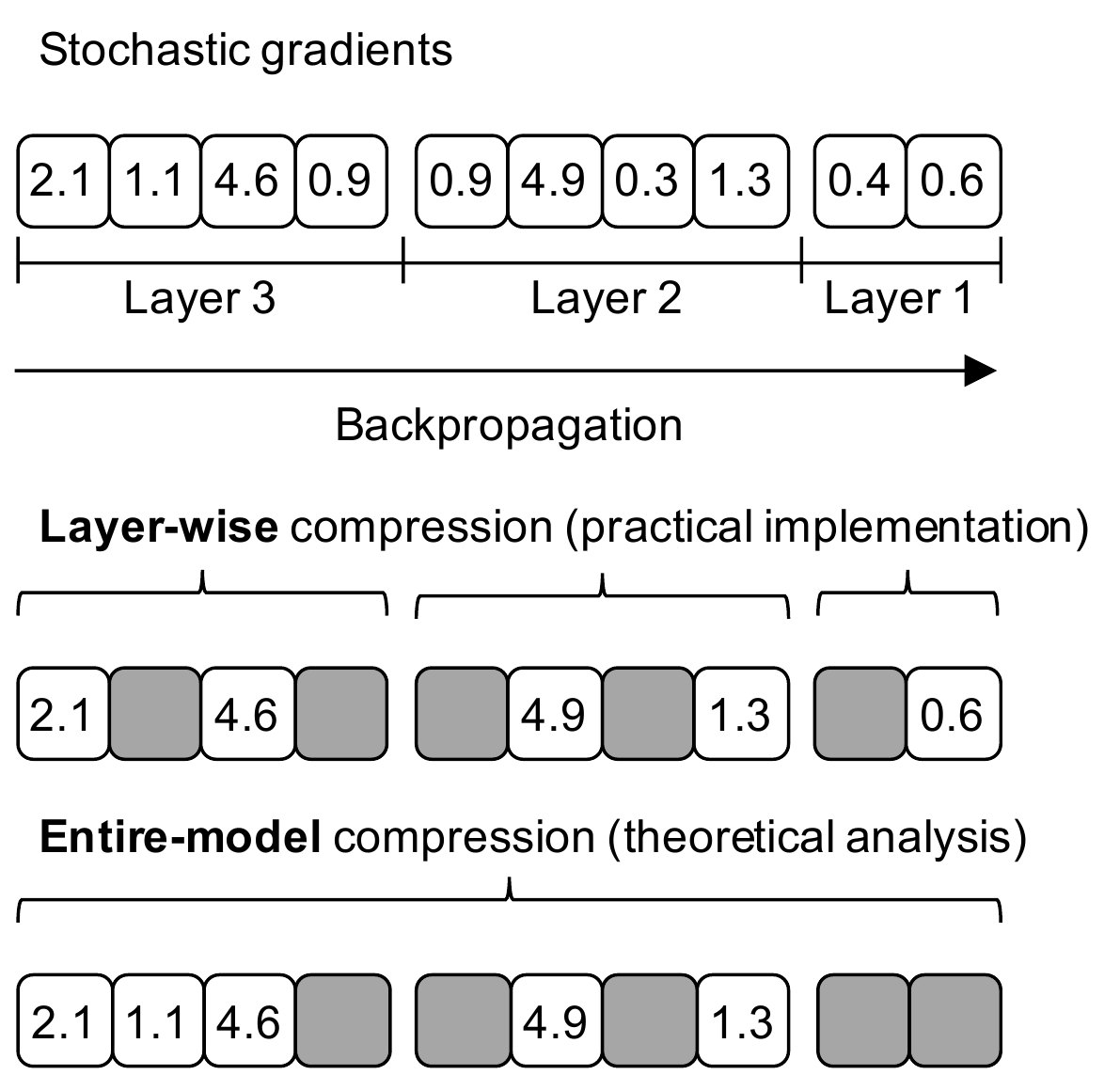}
    \caption{\small{Contrived example of compressed communication illustrating how layer-wise Top $k$ compression (with $k=50\%$) differs from entire-model compression.}}
    \label{fig:layerwise_architec}
\end{figure}

Importantly, layer-wise compression in general differs from entire-model compression (though for certain quantization methods the results are identical). For example, Figure~\ref{fig:layerwise_architec} shows the effects of Top $k$ with a sparsification ratio of 50\%, highlighting that when entire-model compression is used, no gradient for the last layer is transmitted at that specific step, which may affect convergence. This suggests that the choice of a compression approach may warrant careful consideration in practice.

In particular, this discrepancy has important implications that motivate this paper.
First, since implementation artifacts differ from what has been theoretically analyzed, do theoretical guarantees continue to hold?
Second, how does the convergence behavior in the layer-wise compression setting theoretically compare to entire-model compression?
Third, in practice, are there significant differences in terms of convergence behavior when entire-model or layer-wise compression is applied? And if so, how do these differences vary across compression methods, compression ratios, and DNN models?
To the best of our knowledge, this is the first paper to observe the above discrepancy and explore these questions.
To answer these questions, this paper makes the following contributions.

\smartparagraph{Layer-wise bidirectional compression analysis:} 
We introduce a unified theory of convergence analysis for distributed SGD with {\em layer-wise} compressed communication. Our analysis encompasses the majority of existing compression methods and applies to both biased (e.g., Top $k$, Random $k$, signSGD~\cite{signsgd}) and unbiased methods (e.g., QSGD~\cite{alistarh2017qsgd}, TernGrad~\cite{terngrad}, $\cC_{\rm NAT}$~\cite{cnat}). Additionally, our analysis considers {\em bidirectional compression}, that is, compression at both the worker side and parameter server side, mimicking the bidirectional
strategy used in several compression methods~\cite{signsgd,cnat}. Our theoretical analysis gives a proof of {\em tighter convergence bounds} for layer-wise compression as compared to entire-model compression.

\smartparagraph{Evaluation on standard benchmarks:}
We confirm our analytical findings by empirically evaluating a variety of compression methods (Random $k$, Top $k$, TernGrad, Adaptive Threshold, Threshold $v$, and QSGD) with standard CNN benchmarks~\cite{DAWNBench} for a range of models (AlexNet, ResNet-9, and ResNet-50) and datasets (CIFAR-10 and ImageNet). We mainly find that, in many cases, layer-wise compression is better or comparable to entire-model compression in terms of test accuracy at model convergence. However, despite the theoretical findings, our empirical results reveal that in practice there are cases, such as the Top $k$ method with small sparsification ratio $k$ and small model sizes, where layer-wise compression performs worse than entire-model compression. This suggests that the current practice of implementing compression methods in a layer-wise fashion out of implementation expedience may not be optimal in all cases. Thus, it would be advantageous for distributed deep learning toolkits to include support for both layer-wise and entire-model compression.

\section{Preliminaries}

Distributed DNN training builds on the following optimization problem:
\begin{eqnarray}\label{prblm:main_prblm}
\textstyle \min_{x\in \R^d}f(x)\eqdef\frac{1}{n}\min_{x\in \R^d}\sum_{i=1}^n\underbrace{\mathbb{E}_{\xi\sim\cD_i}F_i(x,\xi)}_{:= f_i(x)}. \end{eqnarray}
 Without loss of generality, consider the above problem as a classical empirical risk minimization problem over $n$ workers, where $\cD_i$ is the local data distribution for worker $i$, $\xi$ is a random variable referring to a sample data. These problems typically arise in deep neural network training in the synchronous data-parallel distributed setting, where each worker has a local copy of the DNN model. Each worker uses one of $n$ non-intersecting partitions of the data, $D_i$, to jointly update the model parameters $x\in\R^d$, typically the weights and biases of a DNN model. In this paradigm, the objective function $f(x)$ is non-convex but has Lipschitz-continuous gradient. One of the most popular algorithms for solving (\ref{prblm:main_prblm}) is the stochastic gradient descent (SGD) algorithm~\cite{robins_monro}. For a sequence of iterates $\{x_k\}_{k\ge 0}$ and a step-size parameter $\eta_k>0$ (also called learning rate), SGD iterates are of the form: 
$x_{k+1} = x_k - \eta_kg(x_k)$, where $g(x_k)$ is an unbiased estimator of the gradient of $f$, that is, for a given $x_k$ we have $\textstyle{\mathbb{E}(g(x_k))=\nabla f(x_k).}$ 

\smartparagraph{Notations.}
We write the matrices in bold uppercase letters and denote vectors and scalars by simple lowercase letters. We denote a vector norm of $x\in\R^d$ by $\|x\|$, the $\ell_1$-norm by $\|x\|_1$, the $\ell_2$-norm by $\|x\|_2$, and for a positive definite matrix $\mM$, we define $\|x\|_\mM \eqdef \sqrt{x^\top \mM x}.$ By $x_k^{i,j}$, we denote a vector that results from the $k^{\rm th}$ iteration, at the $i^{\rm th}$ worker, and represents the $j^{\rm th}$ layer of the deep neural network. Similar notation follows for the stochastic gradients. When $j$ is implied, we simplify the notation as $x_k^{i}$ and vice-versa. Also, by $x_k^{1:n}$ we denote a collection of $n$ vectors $x_k^{i}$, where $i=1,\ldots, n$.
Further, for the ease of notation, denote $f(x_k)=f_k$.

\section{Layer-wise Gradient Compression}
\label{sec:algo}
We define a general bidirectional compression framework that is instantiated via two classes of user-provided functions: (1) a compression operator $Q_W$ at each {\em worker} (which for simplicity, we assume is the same at each worker and for every layer), and (2) a compression operator $Q_M$ at the {\em master} node. The master node abstracts a set of parameter servers. We now introduce the framework and then formalize the setup under which we analyze it. Our analysis follows in the next section.

Conceptually, based on its local copy of the model at step $k$, each worker first computes the local stochastic gradient $g_k^{i,j}$ of each layer $j$ (from 1 to $L$) and then performs layer-wise compression to produce $\tilde{{g}}_{k}^{i,j} = Q_{W|j}(g_k^{i,j})$. After that, each worker transmits $\tilde{{g}}_{k}^{i,j}$ to the master. The master collects all the gradients from the workers, aggregates them (via averaging), and then uses the compression operator $Q_M$ to generate $\textstyle{\tilde{g}_k^j\eqdef Q_{M|j}(\frac{1}{n}\sum_{i=1}^n \tilde{{g}}_{k}^{i,j}})$. The master then broadcasts the results back to all workers. Each worker recovers the entire-model gradient $\tilde{g}_k$ by collating the aggregated gradient of each layer $\tilde{g}_k^j$ and then updates the model parameters via the following rule (where $\eta$ is the learning rate): 
\begin{eqnarray}\label{iter:sgd}
\textstyle x_{k+1}=x_k-\eta_k\tilde{g}_k.
\end{eqnarray} 
This process continues until convergence.
Algorithm~\ref{alg_1} lists the steps of this process. 

\begin{algorithm}[tb]
\caption{Layer-wise gradient compression framework}
\label{alg_1}
\textbf{Input:} Number of workers $n$, learning rate $\eta$, compression operators $Q_W$ (worker side) and $Q_M$ (master side)\\
\textbf{Output:} The trained model $x$
\begin{algorithmic}[1] 
\STATE \textbf{On} each worker $i$:
\FOR{$k = 0,1,\ldots$}
\STATE \textbf{Calculate} stochastic gradient ${g}_{k}^{i,j}$ of each layer $j$
\STATE $\tilde{{g}}_{k}^{i,j} = Q_{W|j}(g_{k}^{i,j})$
\STATE \textbf{Send} compressed gradient $\tilde{{g}}_{k}^{i,j}$
\STATE \textbf{Receive} aggregated gradient $\tilde{{g}}_k^j$
\STATE \textbf{Collate} entire-model gradient $\tilde{{g}}_k = \tilde{{g}}_k^{1:L}$
\STATE ${x}_{k+1} = {x}_k - \eta_k\tilde{{g}}_k$
\ENDFOR
\RETURN ${x}$
\end{algorithmic}
\begin{algorithmic}[1]
\STATE \textbf{On} master node, at each step $k$ and for each layer $j$:
\STATE \textbf{Receive} $\tilde{{g}}_k^{i,j}$ from each worker
\STATE $\tilde{{g}}_k^j = Q_{M|j}(\frac{1}{n}\sum_{i=1}^n\tilde{{g}}_k^{i,j})$
\STATE \textbf{Broadcast} aggregated compressed gradient $\tilde{{g}}_k^j$
\end{algorithmic}
\end{algorithm}

We note that this framework is agnostic to the optimization algorithm. We consider SGD in this paper. However, given access to $x_k^{i,j}$ and $\tilde{g}_k^j$, Algorithm \ref{alg_1} can be adapted to any other popular optimizer used to train DNNs, such as ADAM~\cite{adam}, ADAGrad~\cite{DBLP:journals/jmlr/DuchiHS11} or RMSProp.

Moreover, the framework supports different compression operators at the worker side and master side as our general theory supports it. In the limit, the compression operator may also differ between layers, including the identity function as an operator for specific layers to avoid compressing those. This is also covered by our theory.

Finally, while we cast our framework on the parameter-server architecture, it is easy to see that it generalizes to collective routines (specifically, {\tt all\_reduce}) since in that case, there is no master and this behavior is modeled by taking $Q_M$ as the identity function.

\subsection{Setup}
\label{sec:setup} 
We now formalize the above concepts and state
the general assumptions we make (several of which are classical ones).
\begin{assumption}(Lower bound) 
The function $f$ is lower bounded; that is, there exists an $f^\star \in\R$ such that $f(x)\ge f^\star$, for all $x$. 
\end{assumption}
\begin{assumption}($\cL $-smoothness) \label{assm:smoothness}
The function $f$ is $\cL $ smooth if its gradient is $\cL $-Lipschitz continuous, that is, for all $x,y\in\R^d$, $\textstyle{\|\nabla f(x)-\nabla f(y)\|\le \cL \|x-y\|}$.  
\end{assumption}
\begin{assumption}(Unbiasedness of stochastic gradient)\label{assum:Unbiased_grad}
The stochastic gradient is unbiased, that is,
\begin{eqnarray}
\textstyle\mathbb{E}_{}(g_k|x_k)=\nabla f_k.
\end{eqnarray} 
\end{assumption} 
If one assumes that the stochastic gradient has bounded variance denoted as $\Sigma$, then, for a given symmetric positive definite (SPD) matrix $\mA$, one has 
\begin{eqnarray*}
\textstyle \mathbb{E}(\|g_k\|_{\bf A}^2|x_k)={\rm Trace}(\mA \Sigma)+\|\nabla f_k\|_{\bf A}^2,
\end{eqnarray*} 
where ${\rm Trace}(\mX)$ denotes the sum of the diagonal elements of a matrix $\mX$. 
A relaxed assumption of the bounded variance is the {\em strong growth condition on stochastic gradient}.
\begin{assumption}(Strong growth condition on stochastic gradient)\label{assm:bdd_var}
For a given SPD matrix $\mA$, a general strong growth condition with an additive error is 
\begin{eqnarray}\label{eq:gen_strong_growth_cond}
\textstyle \mathbb{E}(\|g_k\|_{\bf A}^2|x_k)\le \rho\|\nabla f_k\|_{\bf A}^2+\sigma^2,
\end{eqnarray}
where $\rho>0$ and $\sigma>0$.
\end{assumption}
A similar assumption was proposed in \cite{vaswani2018fast,survey_largescaleML,bertsekas} when $\mA$ is the identity matrix, that is, for the $\ell_2$-norm. For overparameterized models such as DNNs, it is common practice to assume $\sigma=0$; and the condition says that the growth of stochastic gradients is {\em relatively} controlled by the gradient $\nabla f_k$~\cite{vaswani2018fast}. That is, there exists a $\rho>0$ such that 
\begin{eqnarray*}
 \mathbb{E}(\|g_k\|_{\bf A}^2)\le \rho\|\nabla f_k\|_{\bf A}^2. 
\end{eqnarray*}

Before defining compression operators formally, below we introduce an assumption that compressor operators should obey. Consider a compression operator $Q(\cdot):\R^d\to\R^d$.
\begin{assumption}(Compression operator) \label{assum:compression}
For all vectors $x\in\R^d$ 
the compression operator $Q(\cdot)$ satisfies 
\begin{eqnarray}\label{eq:kcontrac}
\textstyle \mathbb{E}_{Q}\|Q(x)\|_2^2\le (1+\Omega)\|x\|_2^2.
\end{eqnarray}
where the expectation $\mathbb{E}_{Q}(\cdot)$ is taken over the internal randomness of the operator $Q(\cdot)$ and $\Omega >0.$
\end{assumption}
\begin{remark}
A broad range of compression operators, whether biased or unbiased, satisfy Assumption \ref{assum:compression}.
In particular, existing compression operators such as Random $k$, Top $k$, signSGD, unbiased Random $k$, QSGD, $\cC_{\rm NAT}$, TernGrad, stochastic rounding, adaptive compression, respect this assumption.
Moreover, if there is no compression, then $\Omega=0.$
\end{remark}

Following \cite{stochastic_threepoint}, we generalize their assumption (c.f. Assumption 3.3) that imposes a descent property to the stochastic gradient.
This assumption lower bounds the expected inner product of the stochastic gradient $\tilde{g}_k$ with the gradient $\nabla f_k$ with a positive quantity depending on a power of the gradient norm while allowing a small residual on the lower bound.
\begin{assumption}\label{assum:decrease-second}
There exists $0<\alpha \le 2$ such that 
\begin{equation}  
\textstyle \mathbb{E} \left[\tilde{g}_k^\top \nabla f_k\right] \ge \mathbb{E} \|\nabla f_k\|^\alpha +  R_k,
\end{equation}
where $\|\cdot\|$ is a vector norm in $\R^d$ and $R_k$ is a small scalar residual which may appear due to the numerical inexactness of some operators or due to other computational overheads. 
\end{assumption}
 By setting $\alpha=1$ and $R_k = 0$, we recover the key assumption made by \cite{stochastic_threepoint}.

In light of our framework and the assumptions made, we now define a general {\em layer-wise compression operator}.
\begin{definition}({Layer-wise compression operator.}) Let $Q(\cdot):\R^d\to\R^d$ be a layer-wise compression operator such that $Q\eqdef (Q_1\;Q_2\cdots Q_L),$ where each $Q_j(\cdot):\R^{d_j}\to\R^{d_j}$, for $j=1,2,\cdots, L$ with $\sum_{j=1}^Ld_j=d$ be a compression operator. 
\end{definition}
The following lemma characterizes the compression made by biased layer-wise compression operators. 
\begin{lemma}\label{lemma:layer_comp}
Let $Q(\cdot):\R^d\to\R^d$ be layer-wise biased compression operator with $Q\eqdef (Q_1\;Q_2\cdots Q_L)$, such that, each $Q_j(\cdot):\R^{d_j}\to\R^{d_j}$ for $j=1,2,\cdots, L$ satisfies Assumption \ref{assum:compression} with $\Omega=\Omega_j$. Then we have 
\begin{eqnarray}\label{eq:q_sum}
\textstyle\mathbb{E}_Q\left(\|Q(x)\|_2^2\right)&\le&\sum_{1\le j\le L}(1+\Omega_j)\|x^j\|_2^2\nonumber\\&\le&\max_{1\le j\le L}(1+\Omega_j)\|x\|_2^2.
\end{eqnarray}
\end{lemma}
With the above general setup, we next establish convergence of Algorithm \ref{alg_1}.

\vspace{-6pt}
\section{Convergence Analysis}
\label{sec:conv} 
We now establish the convergence of the above-defined layer-wise bidirectional compression scheme. The proofs are available in Appendix~\ref{app:proofs}.
Let the matrix $\textstyle{{\bf W}_W\eqdef{\rm diag}((1+\Omega_W^1)\mI_1\;(1+\Omega_W^2)\mI_2\cdots (1+\Omega_W^L)\mI_L )}$ be a diagonal matrix that characterizes the layer-wise compression at each {\em worker}, such that for each $j=1,2,\cdots, L$, $\mI_j$ be a $d_j\times d_j$ identity matrix. Similarly, to characterize the layer-wise compression at the {\em master} node, we define $\textstyle {{\bf W}_M}$. Given that $\Omega_W^j, \Omega_M^j\ge0$ for each $j=1,2,\cdots, L$, therefore, ${\bf W}_W$ and ${\bf W}_M$ are SPD matrices. Denote $\Omega_M\eqdef \max_{1\le j\le L} \Omega_M^j$, $\Omega_W\eqdef \max_{1\le j\le L} \Omega_W^j$. Further define $\textstyle{\mA\eqdef\mW_M \mW_W}$.

In the next lemma, we consider several compression operators that satisfy Assumption \ref{assum:decrease-second}. For instance, these include unbiased compression operators, as well as Random $k$ and signSGD. 
\begin{lemma}\label{lemma:characterize_compression}
We note the following:
\begin{enumerate}[i.]
\item (For unbiased compression) If $\tilde{g}_k$ is unbiased (the case when $Q_M$ and $Q_W$ are unbiased), then
\begin{equation}
\mathbb{E} \left[\tilde{g}_k^\top \nabla f_k\right] = \mathbb{E} \|\nabla f_k\|_2^2. 
\end{equation}
Therefore, $\tilde{g}_k$ satisfies  Assumption~\ref{assum:decrease-second} with $\alpha=2$, $\|\cdot\|= \|\cdot\|_2$ and $ R_k = 0$.

\item  If $Q_M$ and $Q_W$ are the Random $k$ compression operator with sparsification ratios $k_M$ and $k_W$, respectively, then
\begin{equation}
\textstyle \mathbb{E} \left[\tilde{g}_k^\top \nabla f_k\right] = \frac{k_M k_W}{d^2} \mathbb{E} \|\nabla f_k\|_2^2.
\end{equation}
Therefore, $\tilde{g}_k$ satisfies  Assumption~\ref{assum:decrease-second} with $\alpha=2$, $\|\cdot\|=\frac{k_M k_W}{d^2} \|\cdot\|_2$, and $ R_k = 0$.

\item Let $Q_M$ be the layer-wise Random ${k_{M_j}}$ compression operator for each layer $j$. Similarly, $Q_W$ is the layer-wise Random ${k_{W_j}}$ compression operator for each layer $j$, then
\begin{equation}
\textstyle \mathbb{E} \left[\tilde{g}_k^\top \nabla f_k\right] = \mathbb{E} \|\nabla f_k\|_{\mB}^2, 
\end{equation}
where $\textstyle{\mB= {\rm diag}\left(\frac{k_{M_1} k_{W_1}}{d_1^2}\mI_1\;\frac{k_{M_2} k_{W_2}}{d_2^2}\mI_2\cdots \frac{k_{M_L} k_{W_L}}{d_L^2}\mI_L \right)}$.
Therefore, $\tilde{g}_k$ satisfies  Assumption~\ref{assum:decrease-second} with $\alpha=2$, $\|\cdot\|= \|\cdot\|_{\mB}$ and $ R_k = 0$.

\item  Let $Q_M$ and $Q_W$ be the {\rm sign} function, similar to that in signSGD, then
\begin{equation}
\textstyle \mathbb{E} \left[\tilde{g}_k^\top \nabla f_k\right] \ge \mathbb{E} \|\nabla f_k\|_1 +  R_k. 
\end{equation}
Therefore, $\tilde{g}_k$ satisfies Assumption~\ref{assum:decrease-second} with  $\alpha=1$, $\|\cdot\|= \|\cdot\|_{1}$, and $R_k = \cO\left(\frac{1}{BS}\right)$, where $BS$ is the size of used batch to compute the signSGD.
\end{enumerate}
\end{lemma}

Similar to the cases mentioned in Lemma 2, we can characterize several other well-known compression operators or their combinations. Our next lemma gives an upper bound on the compressed stochastic gradient $\tilde{g}_k$.  
\begin{lemma}\label{lemma:secondorderterm}
Let Assumption \ref{assm:bdd_var} hold. 
With the notations defined above, we have
\begin{equation}\label{eq:secondorderterm}
\textstyle \mathbb{E} \|\tilde{g}_k\|_2^2 \le \rho \|\nabla f_k\|_{\bf A}^2 + \sigma^2.
\end{equation}
\end{lemma}

\begin{remark}
  If $g_k^i$ has bounded variance, $\Sigma$, say, then $\textstyle{{\rm Trace}(\mA\Sigma)=\sigma^2.}$
\end{remark}
Now we quote our first general inequality that the iterates of \eqref{iter:sgd} satisfy. This inequality does not directly yield convergence of the scheme in \eqref{iter:sgd}. However, this is a first necessary step to show convergence. We note that the matrix $\mA$ and $\sigma$ quantify layer-wise compression. 
\begin{proposition}
\label{theorem:recurrence}
With the notations and the framework defined before, the iterates of \eqref{iter:sgd} satisfy 
\begin{eqnarray}\label{eq:main-result}
\textstyle \eta_k\left(\mathbb{E}\|\nabla f_k\|^\alpha - \frac{\cL\eta_k}{2}\mathbb{E}\|\nabla f_k\|_{\bf A}^2 \right)&\leq  \mathbb{E}(f_k-f_{k+1}) \\&-\eta_k R_k\nonumber +\tfrac{\cL\eta_k^2\sigma^2}{2}.
\end{eqnarray}
\end{proposition}
\begin{remark}
If $\tilde{g}_k$ is an unbiased estimator of the gradient, then $\alpha =2$,  $\|\cdot\|= \|\cdot\|_2$, $\mA=\mI$, and $R_k = 0$. Therefore, (\ref{eq:main-result}) becomes
\begin{equation*}\label{eq:main-result_sgd}
\textstyle \eta_k\mathbb{E}\|\nabla f_k\|_2^2\left(1-\tfrac{\cL\rho\eta_k}{2}\right)\leq \mathbb{E}(f_k-f_{k+1}) +\frac{\cL\eta_k^2\sigma^2}{2}.
\end{equation*}
The above is the classic inequality used in analyzing SGD.
\end{remark}

In the non-convex setting, it is standard to show that over the iterations the quantity $\min_{k\in[K]}\mathbb{E}(\|\nabla f_k\|^2)$ approaches to 0 as $K\to\infty.$ Admittedly, this is a weak statement as it only guarantees that an algorithm converges to a local minimum of $f$. To facilitate this, next we quote two propositions: one is for the special case when $\alpha=2$; the other one covers all cases $\alpha\in(0,2]$. In the propositions, for simplicity, we use a fixed step-size $\eta$. One can easily derive the convergence of Algorithm \ref{alg_1} under general compression $Q$ to the $\epsilon$-optimum by choosing a sufficiently small or decreasing step-size, similarly to the classical analysis of SGD. 
\begin{proposition}\label{prop:speciallcase}(Special case.)
Consider $\alpha=2$, and $ R_k = 0$. Let $C>0$ be the constant due to the equivalence between the norms  $\|\cdot\|$ and $\|\cdot\|_{\mA}$, and $K>0$ be the number of iterations. If $\eta_k=\eta = \cO\left(\tfrac{1}{\sqrt{K}}\right)<\tfrac{2}{LC \rho}$ then 
\begin{eqnarray*}
 \textstyle\tfrac{\sum_{k=1}^K \mathbb{E}\|\nabla f_k\|^2}{K}\leq \cO\left(\tfrac{1}{\sqrt{K}}\right). 
\end{eqnarray*}
\end{proposition}

\begin{proposition}\label{prop:generalcase}(General case.) {\it Assume $\|\nabla f_k\|\le G$. Let $C>0$ be the constant coming from the equivalence between  $\|\cdot\|$ and $\|\cdot\|_{\mA}$. Let $\eta_k=\eta=\cO\left(\frac{1}{\sqrt{K}}\right)<\frac{2}{\cL C\rho G^{2-\alpha}}$. Let $a \eqdef \left(1 -  \frac{\cL C \rho G^{2-\alpha} \eta}{2}\right) > 0$, then
  \begin{eqnarray*}
\textstyle \tfrac{\sum_{k=1}^K \mathbb{E} \|\nabla f_k\|^\alpha}{K} \leq \cO\left(\tfrac{1}{\sqrt{K}}\right) + \tfrac{\sum_{k=1}^K R_k}{a K}.
\end{eqnarray*}}
\end{proposition}
\begin{remark}
 Note that if we assume small residuals such that 
 for all $k$, $\textstyle{R_k = \cO\left(\tfrac{1}{\sqrt{K}}\right)}$, then we have
$\textstyle{\tfrac{\sum_{k=1}^K \mathbb{E} \|\nabla f_k\|^\alpha}{K} \leq \cO\left(\tfrac{1}{\sqrt{K}}\right)}$.
\end{remark}

\begin{remark}
  From the above propositions, immediately one can observe, $\textstyle{K\min_{k\in[K]}\mathbb{E}(\|\nabla f_k\|^\alpha)\le \sum_{k=1}^k\mathbb{E}(\|\nabla f_k\|^\alpha)}$ and hence the above propositions directly imply convergence of the iterative scheme in \eqref{iter:sgd} under layer-wise compression. 
\end{remark}

\begin{remark}
  Note that for all the cases we mentioned in Lemma 2, except for signSGD, we have $\alpha=2$, and $ R_k = 0$, so we are in Proposition \ref{prop:speciallcase}.
  For signSGD, $|R_k|\le \cO(1/BS)$ so if one  uses $BS = \cO(\sqrt{K})$, then we get  $
\textstyle{\tfrac{\sum_{k=1}^K \mathbb{E} \|\nabla f_k\|^{\alpha}}{K} \leq \cO\left(\tfrac{1}{\sqrt{K}}\right).}$
\end{remark}
We note that our convergence analysis can be extended to convex and strongly convex cases. 

\vspace{-8pt}
\paragraph{Layer-wise compression vs entire-model compression.} From our analysis, a natural question arises: can we relate the noise for layer-wise compression with that of entire-model compression? The answer is yes and we give a sharper estimate of the noise bound. From our convergence results, we see that the error for layer-wise compression is proportional to ${\rm Trace}({\bf A})=\sum_{j=1}^L(1+\Omega_M^j)(1+\Omega_W^j)$. This is less than or equal to $L\max_{j}(1+\Omega_M^j)(1+\Omega_W^j)$, which is the error for using bidirectional compression applied to the entire model.

\section{Empirical Study}\label{sec:experiments}

We implement several well-known compression methods and show experimental results contrasting layer-wise with entire-model compression for a range of standard benchmarks.

\subsection{Implementation highlights}
We base our proof-of-concept implementation on PyTorch.\footnote{Available at \url{https://github.com/sands-lab/layer-wise-aaai20}.}
Layer-wise and entire-model compression share the same compression operations. The difference is the inputs used with each invocation of the compressor.
As with other modern toolkits, in PyTorch, the gradients are computed layer-by-layer during the backward pass, starting from the last layer of the model. As such, layer $j$'s gradient is available as soon as it is computed and before backward propagation for layer $j-1$ starts. Distributed training typically exploits this characteristic to accelerate training by overlapping some amount of communication with the still ongoing gradient computation.
In the layer-wise setting, our implementation invokes the compressor independently for the gradient of each layer as soon as it is available.
In contrast, in the entire-model setting, our implementation waits until the end of the backward pass and invokes the compressor once with the entire model gradients as input. Clearly, this introduces an additional delay before communication starts; however, with smaller volumes of transmitted data, the benefits of compressed communication can eclipse this performance penalty.

\subsection{Experimental setting} 

\smartparagraph{Compression methods.}
We experiment with the following compression operators: Random $k$, Top $k$, Threshold $v$, TernGrad~\cite{terngrad}, Adaptive Threshold~\cite{adacomp}, and QSGD~\cite{alistarh2017qsgd}. Given an input vector, Random $k$ uniformly samples $k\%$ of its elements; Top $k$ selects the largest $k\%$ elements by magnitude; Threshold $v$ selects any element greater or equal to $v$ in magnitude.

\smartparagraph{Benchmarks.}
We adopt DAWNBench~\cite{DAWNBench} as a benchmark for image classification tasks using convolutional neural networks (CNNs).
We train AlexNet \cite{alexnet}, ResNet-9 and ResNet-50 \cite{resnet152} models.
We use standard CIFAR-10~\cite{cifar10} and ImageNet~\cite{imagenet} datasets.

\smartparagraph{Hyperparameters.} We set the global mini-batch size to 1,024 (the local mini-batch size is equally divided across workers); the learning rate schedule follows a piecewise-linear function that increases the learning rate from 0.0 to 0.4 during the first 5 epochs and then decreases to 0.0 till the last epoch. Unless otherwise noted, CIFAR-10 experiments run for 24 epochs and ImageNet experiments for 34 epochs.
Where applicable, we use ratio $k$ in \{0.1, 1, 10, 30, 50\}\%.

\smartparagraph{Environment.} We perform our experiments on server-grade machines running Ubuntu 18.04, Linux 4.15.0-54, CUDA 10.1 and PyTorch 1.2.0a0\_de5a481. The machines are equipped with 16-core 2.6 GHz Intel Xeon Silver 4112 CPU, 512 GB of RAM and 10 Gbps network interface cards. Each machine has an NVIDIA Tesla V100 GPU with 16 GB of memory. We use two machines for CIFAR-10 experiments while we use four machines for ImageNet experiments.

\smartparagraph{Evaluation metrics.} We report the accuracy on a held-out testing dataset evaluated at the end of each epoch during the training process. We compare the test accuracy of layer-wise and entire-model compression.

\begin{figure}[t]
    \centering
    \begin{minipage}{0.48\textwidth}
      \includegraphics[width=1\textwidth]{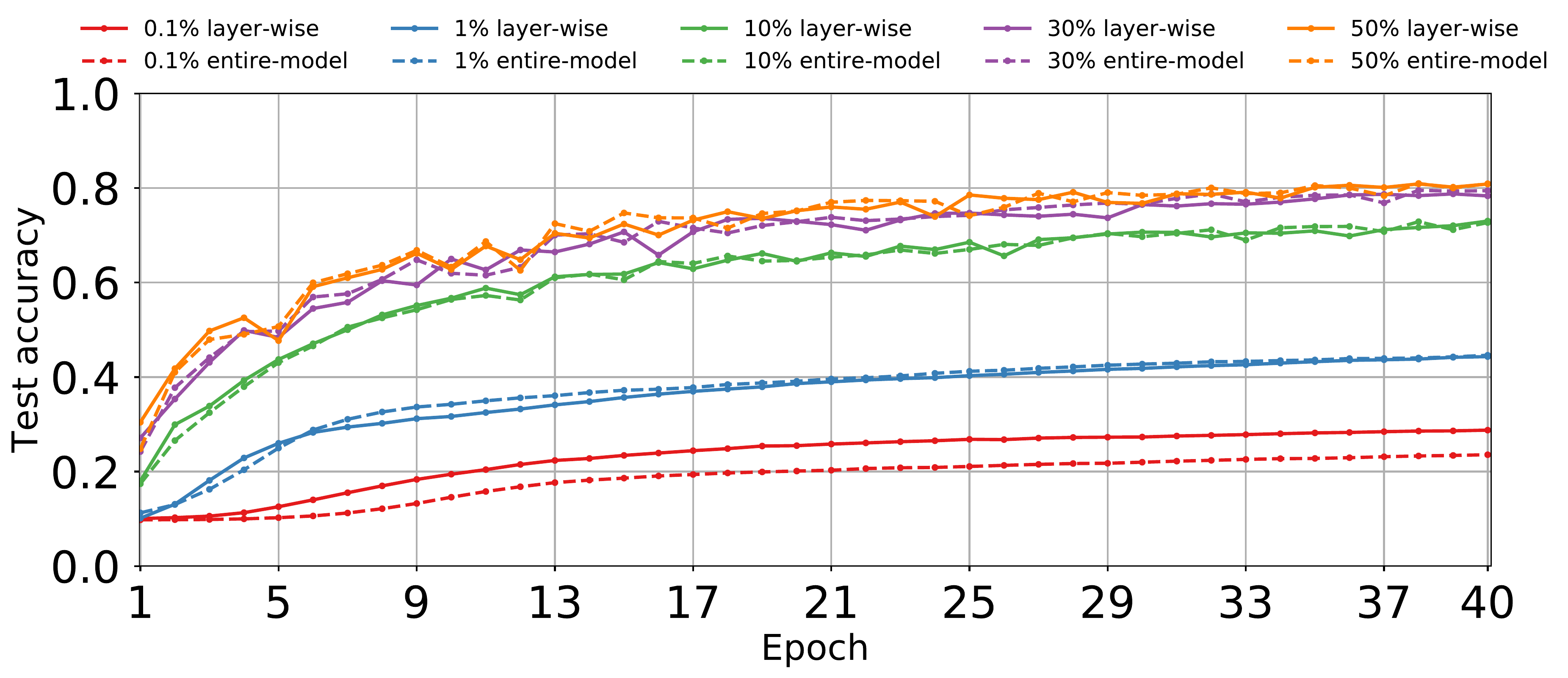}
       \subcaption{\small{AlexNet}}\label{fig:random_alexnet_cifar10}   
     \end{minipage}
     \begin{minipage}{0.48\textwidth}
      \includegraphics[width=1\textwidth]{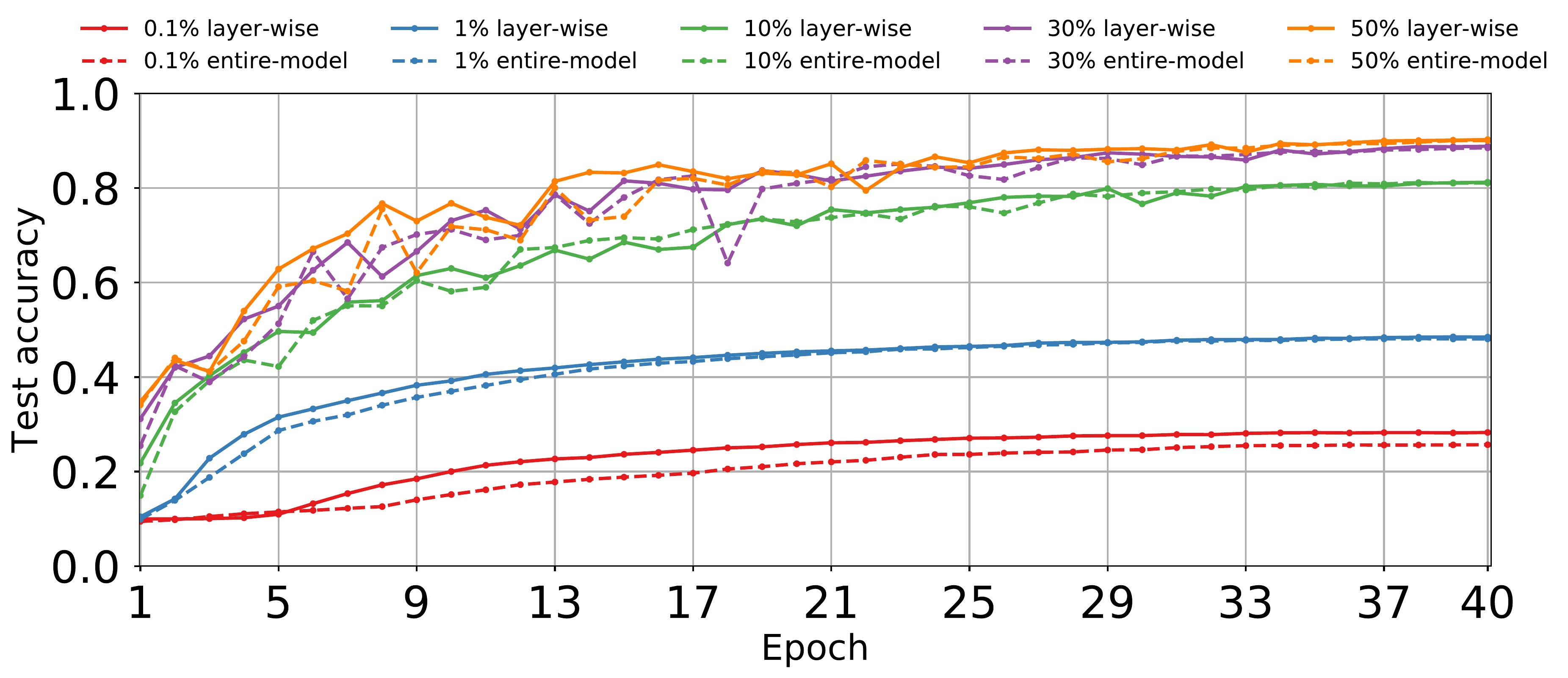}
     \subcaption{\small{ResNet-9}}\label{fig:random_resnet9_cifar10}
     \end{minipage}
     \caption{\small{CIFAR-10 test accuracy for Random $k$ compression.}}
    \label{fig:random_cifar10}
\end{figure}

\subsection{Experimental Results}

Below we illustrate the results for each compression method.
In a nutshell, our results show that both layer-wise and entire-model compression approaches achieve in most cases similar convergence behavior and test accuracy. However, certain compression methods, namely, TernGrad, QSGD, and Adaptive Threshold achieve significantly better accuracy using layer-wise compression. This is because per-layer compression in these cases capitalizes on more fine-grained representations that reduce the overall compression error.

\smartparagraph{Random $k$.}
Figure~\ref{fig:random_cifar10} reports results for Random $k$ compression while training AlexNet and ResNet-9 on CIFAR-10. We observe that layer-wise Random $k$ achieves comparable results to entire-model compression at different sparsification ratios, except for ratio of $0.1\%$ where layer-wise supersedes entire-model compression. This is not surprising because both layer-wise and entire-model compression approaches sample uniformly at random gradient elements, and so, every element has an equal probability of being sampled regardless of its magnitude. We also notice that for ratios less than $10\%$, Random $k$ has a slower rate of convergence for both compression approaches compared to other compression methods (shown below). This suggests that randomly selecting a sample of gradient elements with no regard to their importance is not ideal, especially for small sparsification ratios.

\begin{figure}[t]
    \centering
    \begin{minipage}{0.48\textwidth}
      \includegraphics[width=1\textwidth]{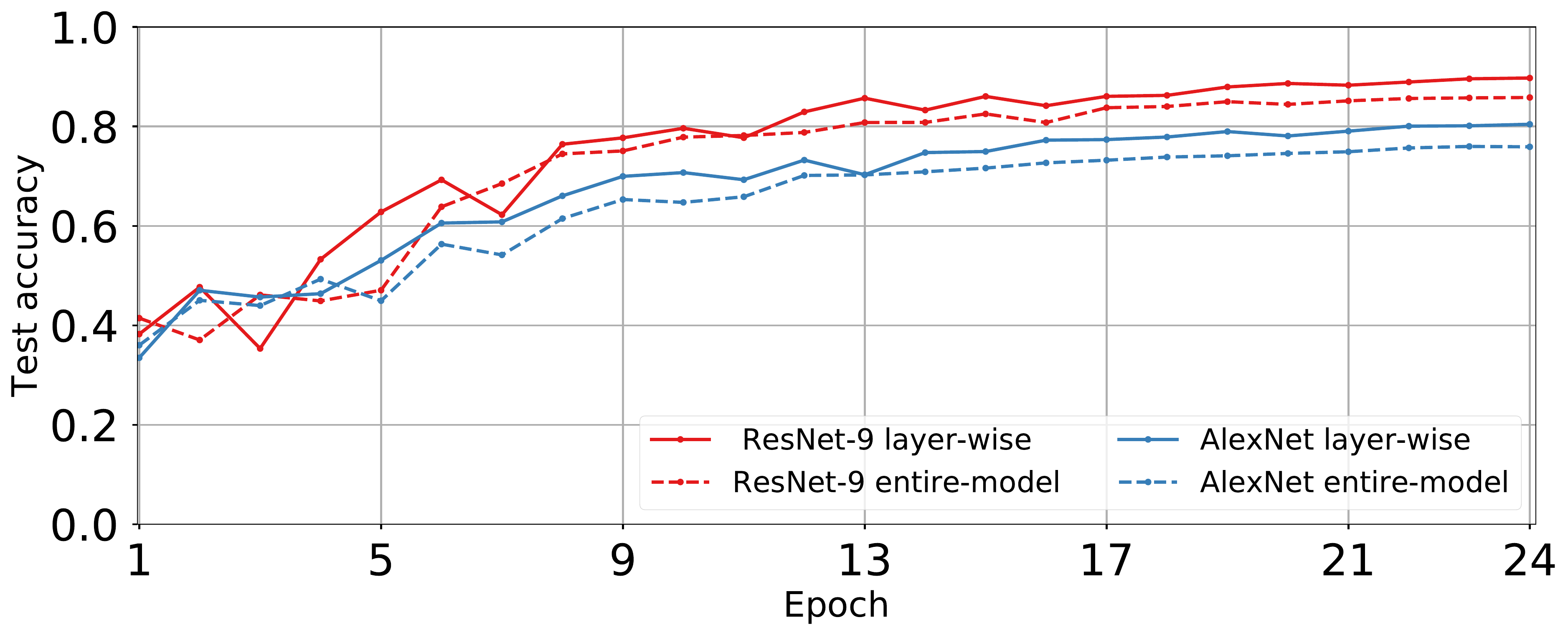}
       \subcaption{\small{CIFAR-10 test accuracy}}\label{fig:terngrad_cifar10}   
     \end{minipage}
     \begin{minipage}{0.48\textwidth}
      \includegraphics[width=1\textwidth]{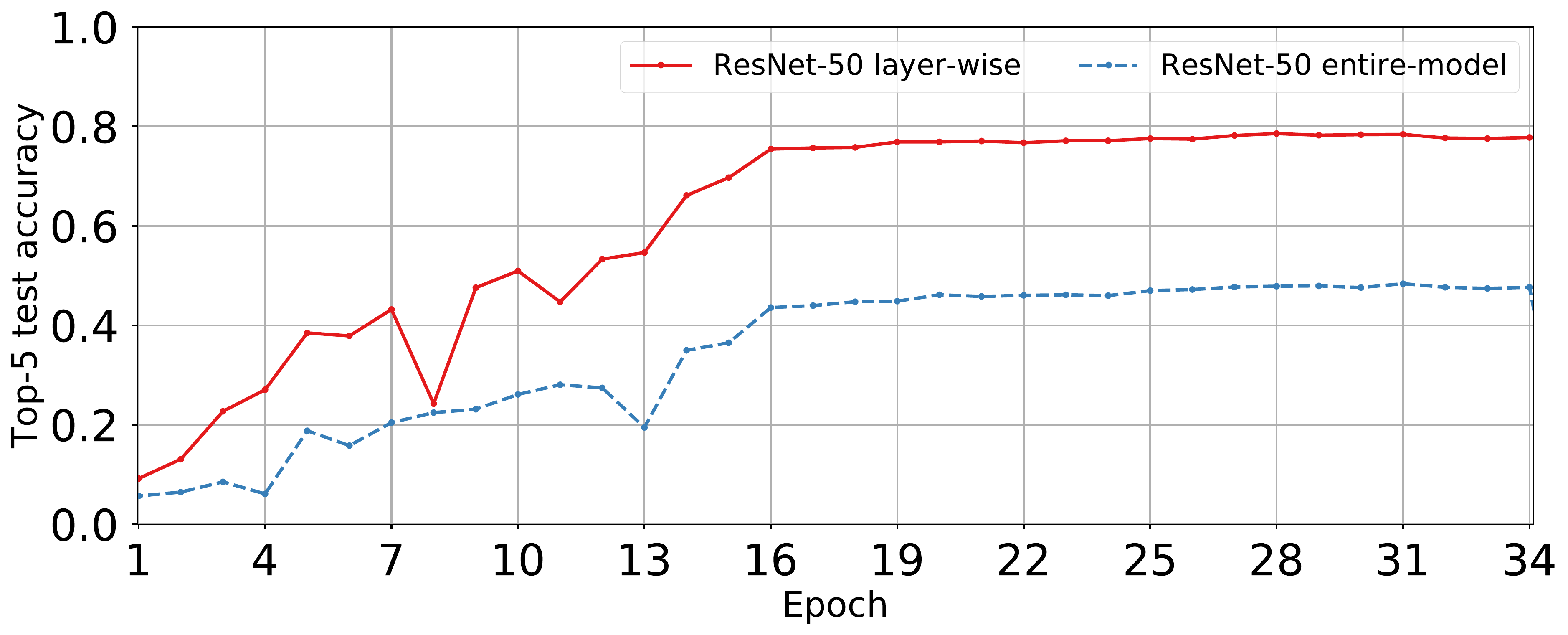}
     \subcaption{\small{ImageNet top-5 test accuracy}}\label{fig:terngrad_imagenet}
     \end{minipage}
     \caption{\small{TernGrad compression.}}
    \label{fig:terngrad}
\end{figure}

\begin{figure}[t!]
    \centering
    \begin{minipage}{0.48\textwidth}
      \includegraphics[width=1\textwidth]{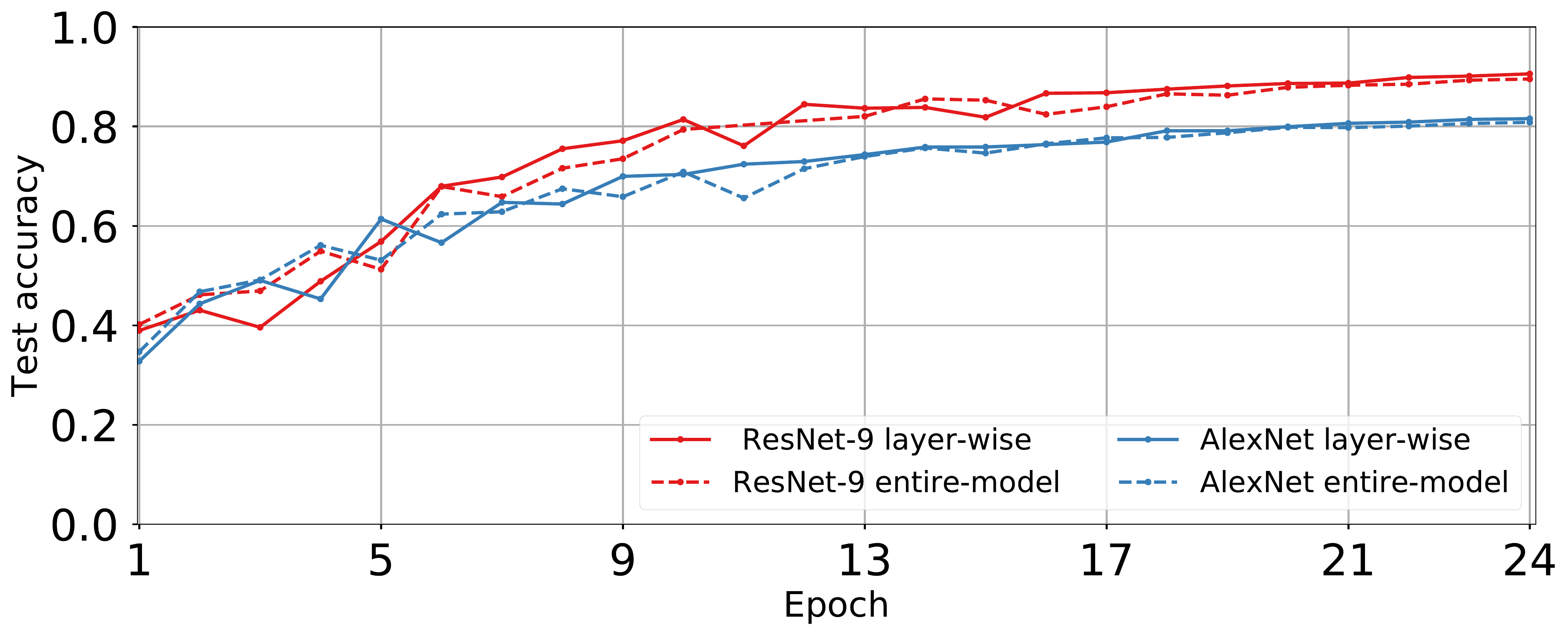}
       \subcaption{\small{CIFAR-10 test accuracy}}\label{fig:qsgd_cifar10}   
     \end{minipage}
     \begin{minipage}{0.48\textwidth}
      \includegraphics[width=1\textwidth]{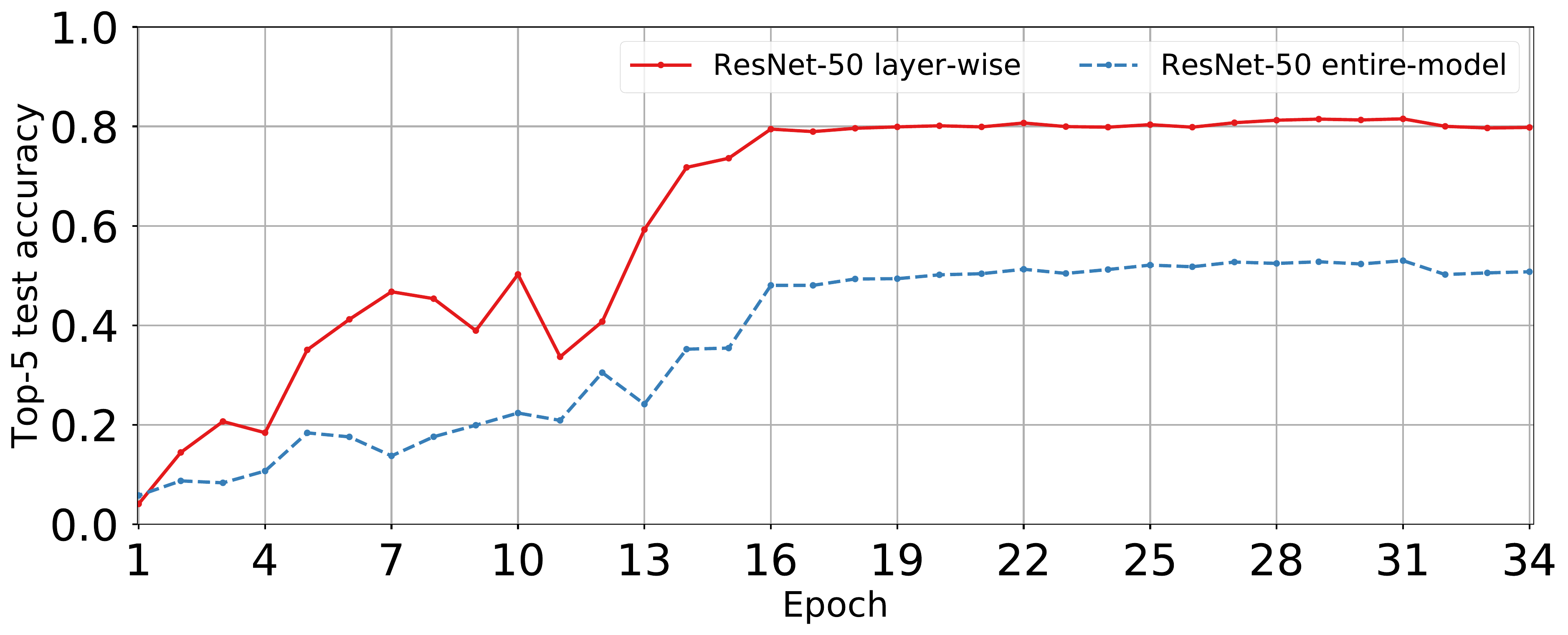}
     \subcaption{\small{ImageNet top-5 test accuracy}}\label{fig:qsgd_imagenet}
     \end{minipage}
    \caption{\small{QSGD compression (with $s=256$).}}
    \label{fig:qsgd}
\end{figure}

\smartparagraph{TernGrad.}
Figure~\ref{fig:terngrad} presents the results of TernGrad compression for several benchmarks. We observe that with TernGrad, layer-wise compression achieves consistently higher test accuracy compared to entire-model compression.
Mostly this result is attributed to the following.
As an unbiased quantization method, TernGrad scales the gradient values (i.e., three numerical levels $\{-1,0,1\}$) by a scalar computed as a function of the gradient vector and its size.
For entire-model compression, there is a single scalar and this may be looser than each layer's scalar used in layer-wise compression.
Thus, when the model is updated (line 8 of Algorithm~\ref{alg_1}), entire-model compression has higher probability of error.

\smartparagraph{QSGD.}
The results using QSGD, shown in Figure~\ref{fig:qsgd} are similar to TernGrad.
We note that ImageNet layer-wise accuracy is 1.52$\times$ better than entire-model compression.

\begin{figure}[t]
    \centering
      \includegraphics[width=0.48\textwidth]{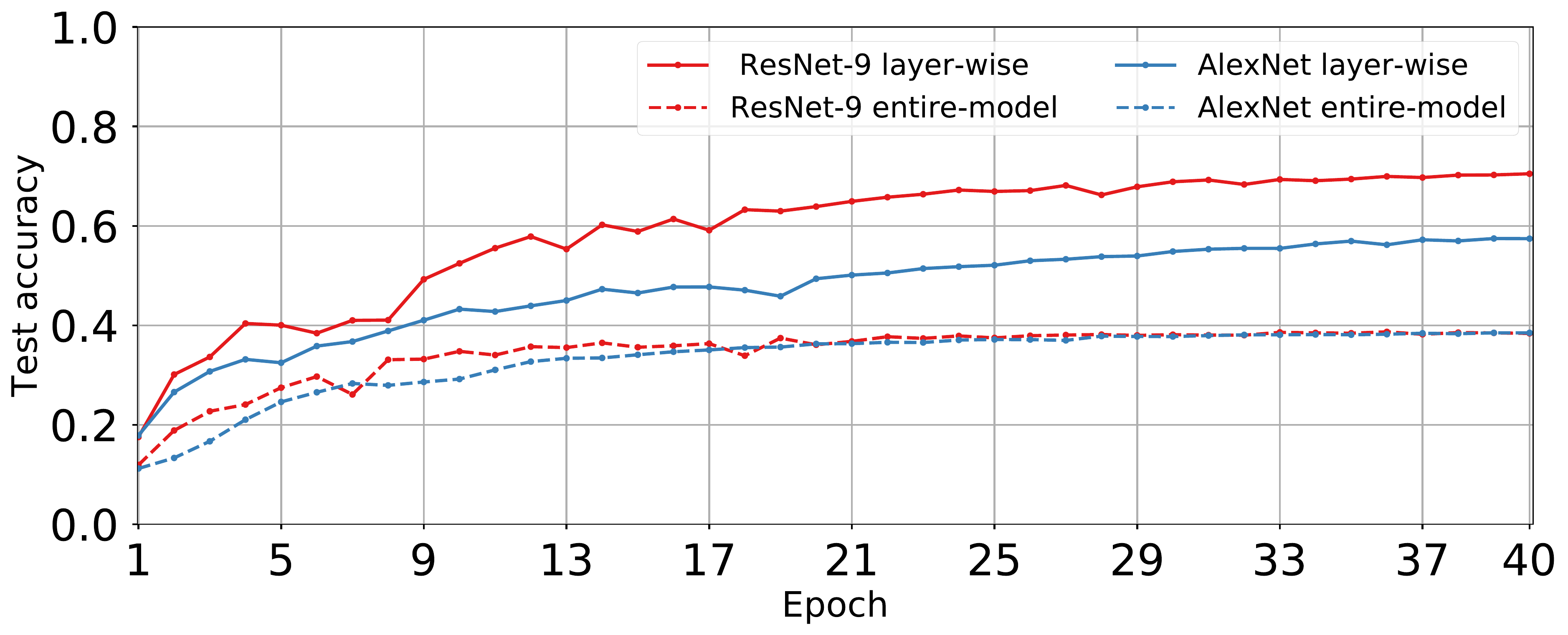}
    \caption{\small{CIFAR-10 test accuracy for Adaptive Threshold.}}
    \label{fig:AdaThresh}
\end{figure}

\begin{figure}[t]
    \centering
      \includegraphics[width=0.48\textwidth]{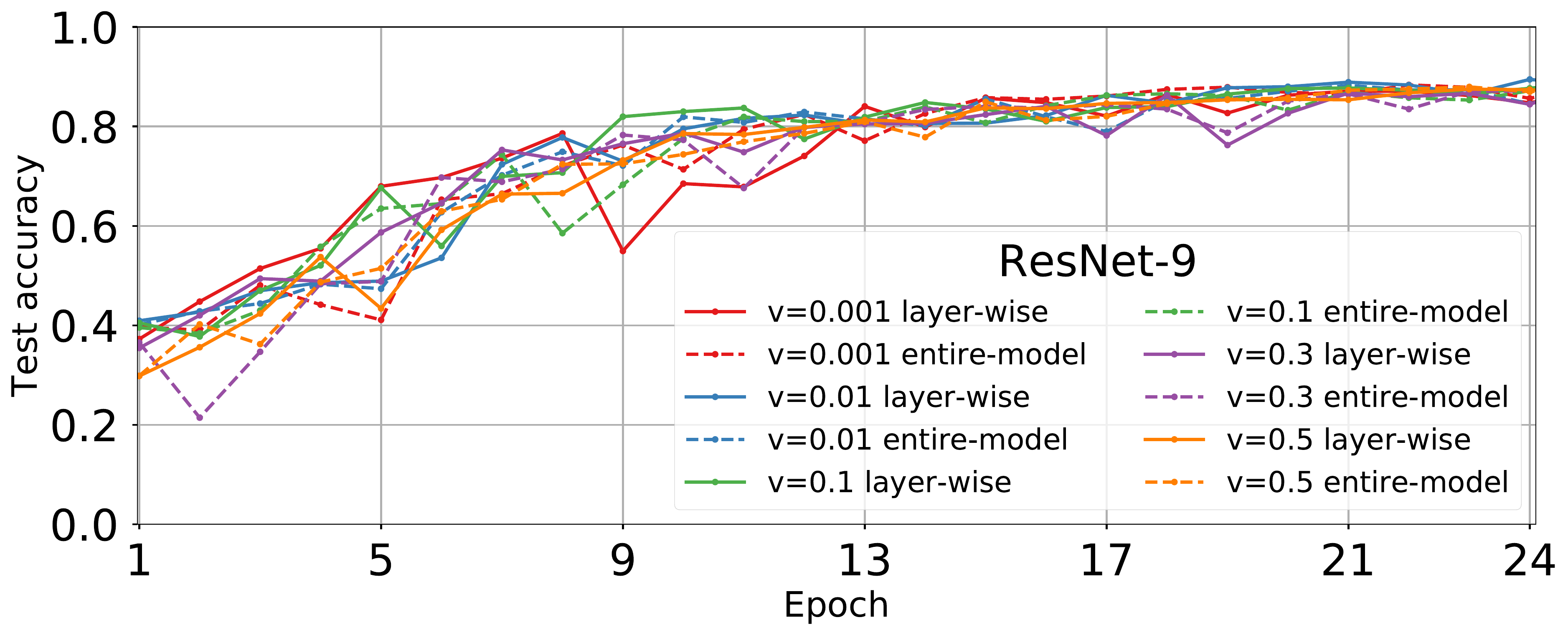}
    \caption{\small{CIFAR-10 test accuracy for Threshold $v$ compression.}}
    \label{fig:thresh_cifar10}
\end{figure}

\begin{figure}[th]
    \centering
    \begin{minipage}{0.48\textwidth}
      \includegraphics[width=1\textwidth]{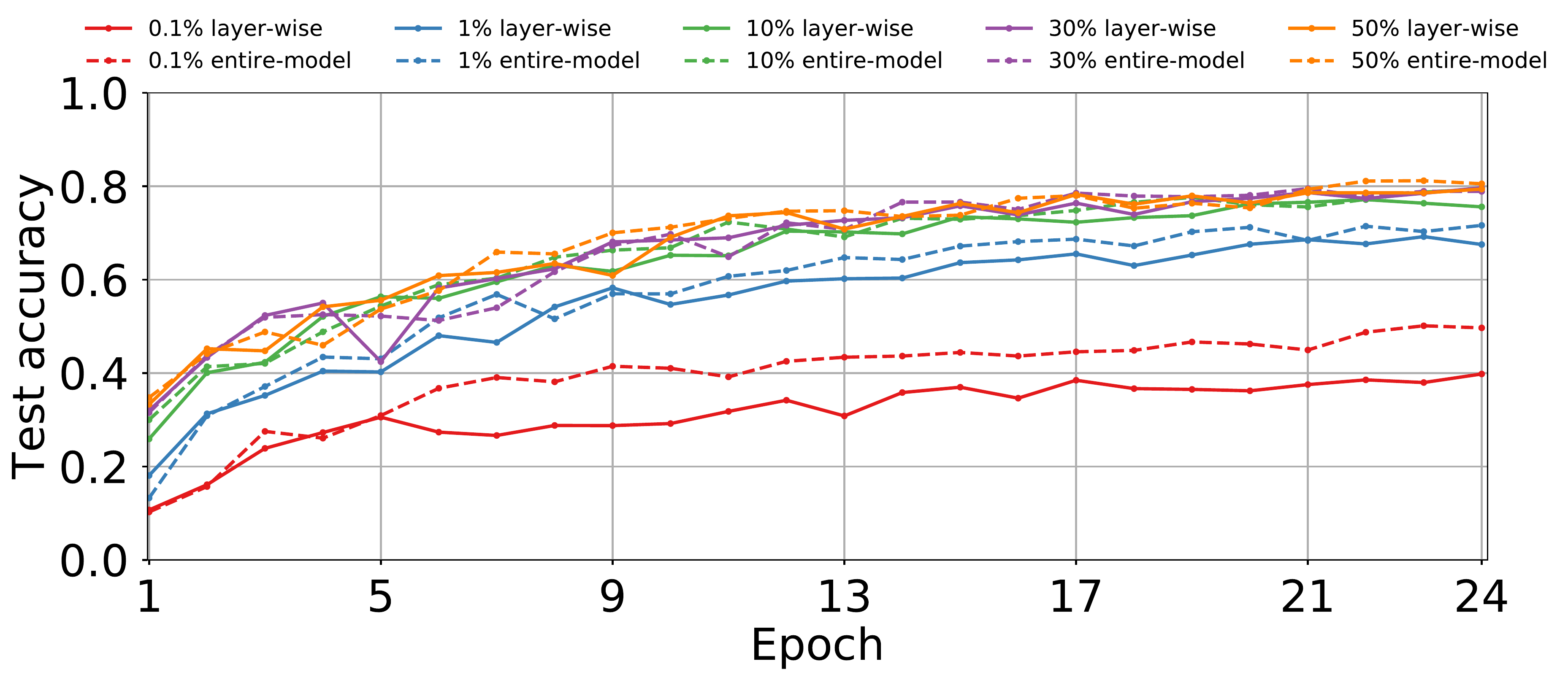}
       \subcaption{\small{AlexNet}}\label{fig:top_alexnet_cifar10}   
     \end{minipage}
     \begin{minipage}{0.48\textwidth}
      \includegraphics[width=1\textwidth]{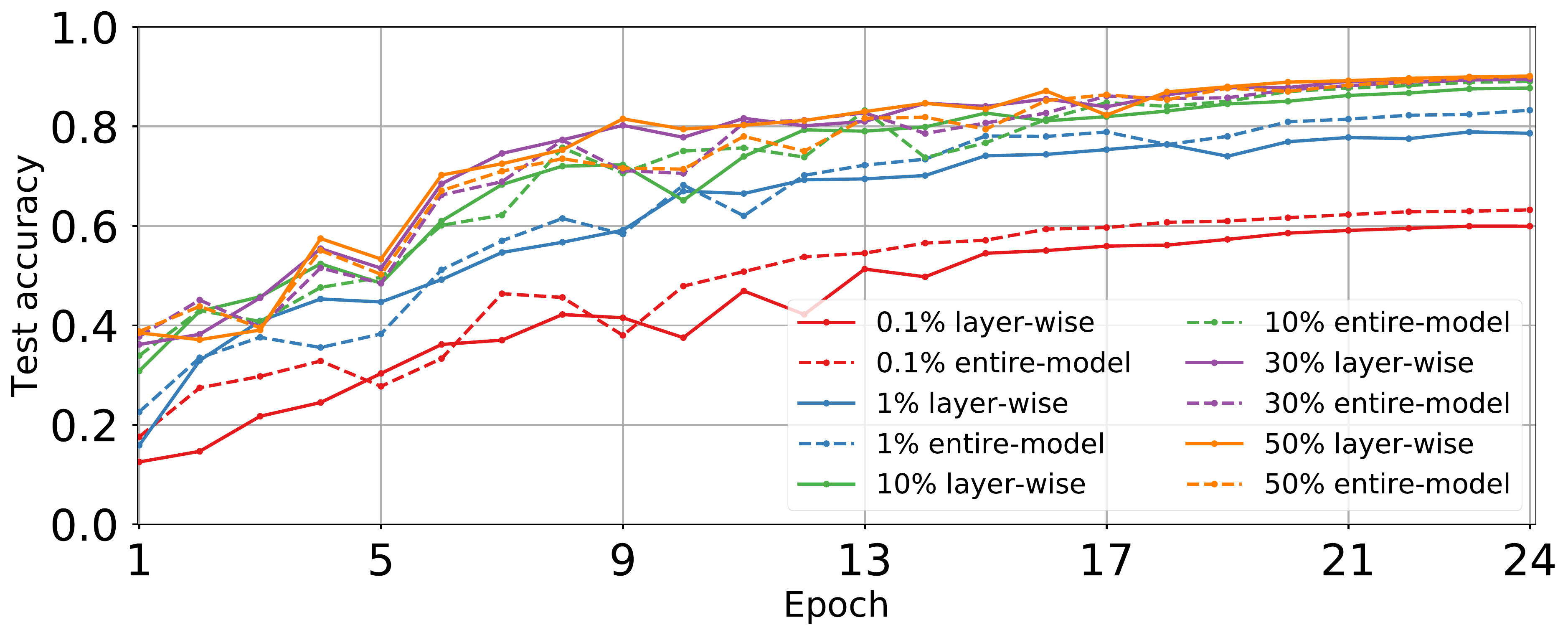}
     \subcaption{\small{ResNet-9}}\label{fig:top_resnet9_cifar10}
     \end{minipage}
     \begin{minipage}{0.48\textwidth}
      \includegraphics[width=1\textwidth]{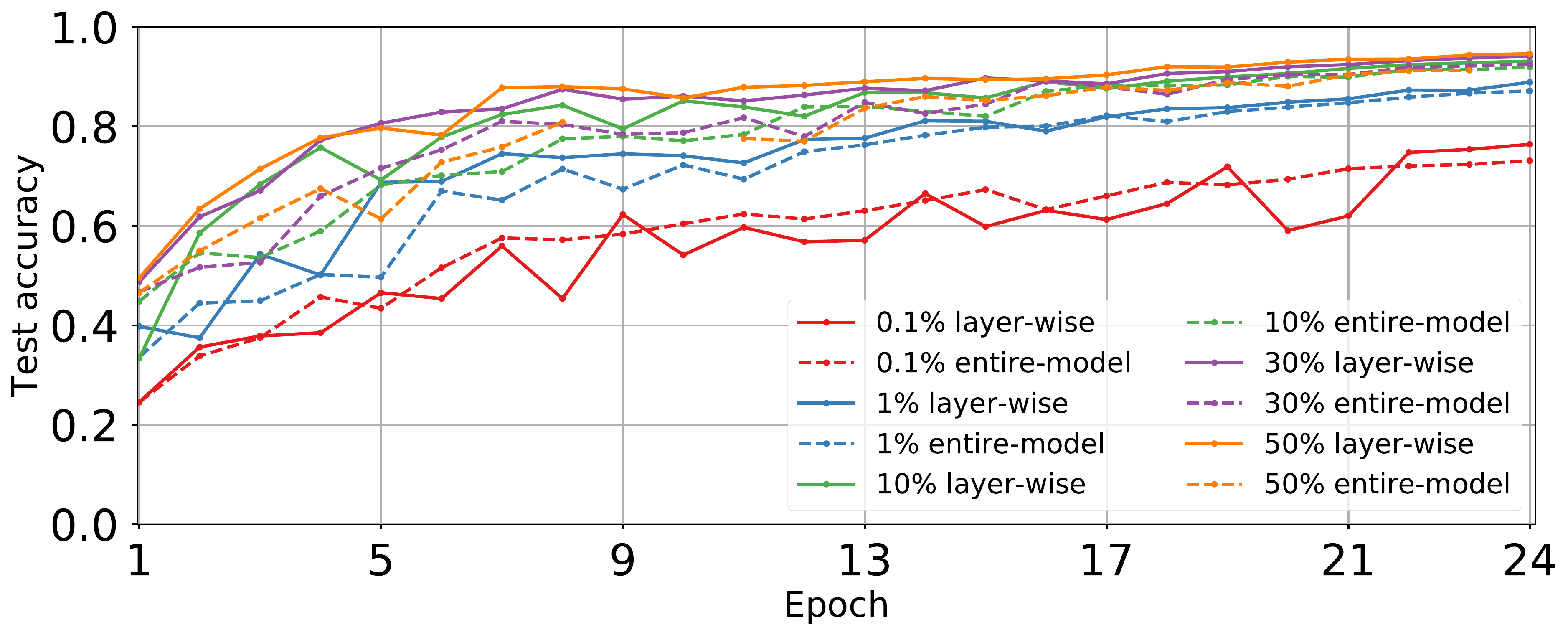}
     \subcaption{\small{ResNet-9 by using SGD with Nesterov's momentum}}\label{fig:top_cifar10_mom}
     \end{minipage}
    \caption{\small{CIFAR-10 test accuracy for Top $k$ compression.}}
    \label{fig:top_cifar10}
\end{figure}

\begin{figure}[t!]
    \centering
      \includegraphics[width=0.48\textwidth]{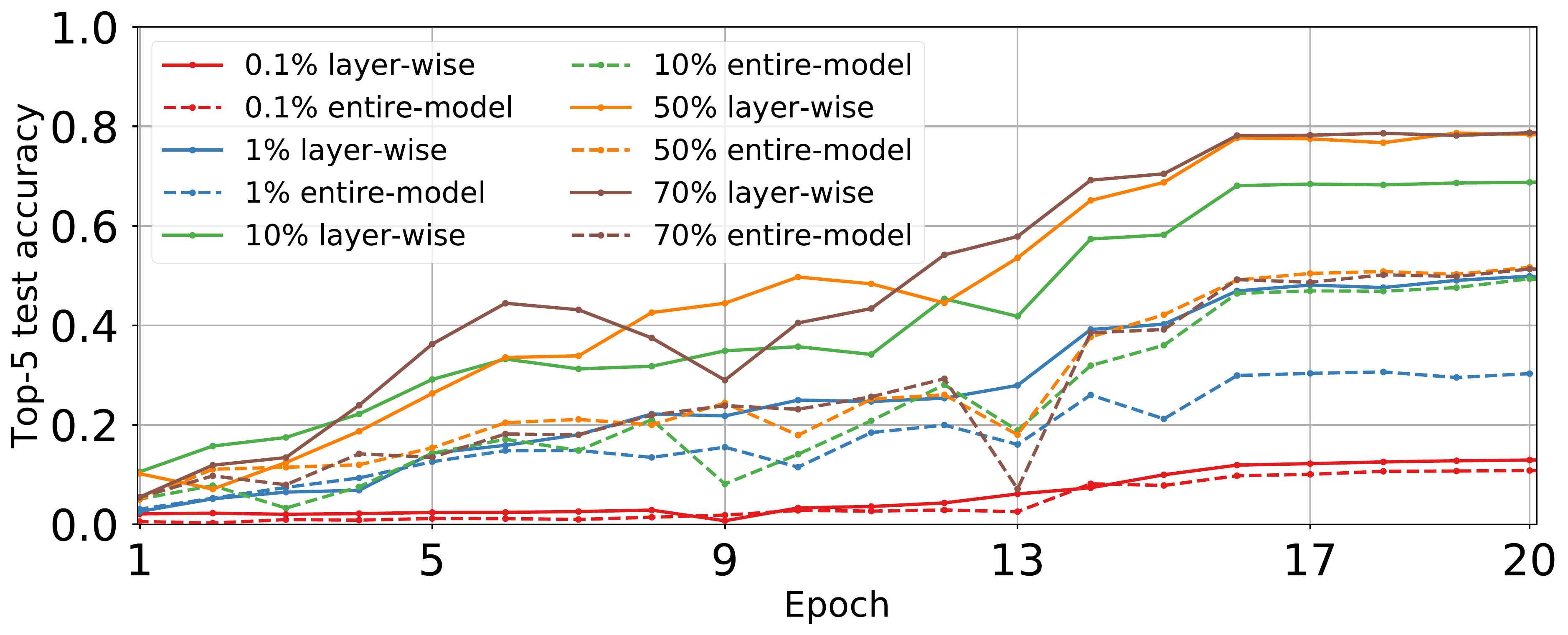}
    \caption{\small{ImageNet top-5 test accuracy for Top $k$ compression.}}
    \label{fig:top_imagenet}
\end{figure}

\smartparagraph{Adaptive Threshold.}
Figure~\ref{fig:AdaThresh} shows the results of Adaptive Threshold compression while training AlexNet and ResNet-9 on CIFAR-10. As before, layer-wise compression achieves better accuracy compared to entire-model compression. The reasoning for this is similar to TernGrad: here, a per-layer threshold chosen with respect to the layer's gradient values performs better than a single threshold selected for the entire-model gradient values. However, we note that this compression method, compared to others, induces slower convergence and only achieves at best $\approx70\%$ accuracy after 40 epochs.

\smartparagraph{Threshold $v$.}
Figures~\ref{fig:thresh_cifar10} reports results of Threshold $v$ compression while training ResNet-9 on CIFAR-10 (AlexNet is qualitatively similar and omitted). We observe that layer-wise and entire-model compression achieve similar accuracy for a wide range of thresholds across three orders of magnitude. This is expected because every gradient element greater than $v$ in magnitude is transmitted in both approaches.

\smartparagraph{Top $k$.}
Figure~\ref{fig:top_cifar10} presents the results on CIFAR-10 for Top $k$ compression.
Similarly to Random $k$, layer-wise compression achieves comparable results to entire-model compression for a range of sparsification ratios. 

Interestingly, for ratios $\leq 10\%$, where test accuracy is overall low, entire-model compression performs better than layer-wise compression.
We could attribute this to the relevance of different layers to the model's convergence, in accordance to a recent study on the relevance of layers in DNNs~\cite{Montavon2018}. Unlike layer-wise compression, the top $k\%$ gradient elements picked by entire-model compression could indeed be more important towards the optimization objective, that is, the loss function. Small sparsification ratios and models with a relatively small number of layers stress this behavior.
However, to highlight that layer-wise compression is not necessarily inferior in these settings, we repeat the experiment training ResNet-9 by using SGD with Nesterov's momentum (which is outside the scope of our theoretical analysis).
Figure~\ref{fig:top_cifar10_mom} shows that layer-wise compression is comparable to, if not better than entire-model compression even at small ratios.
We leave a thorough analysis of these observations to future work.

Figure~\ref{fig:top_imagenet} shows the top-5 test accuracy for the ResNet-50 model trained on ImageNet (only 20 epochs shown for clarity). Layer-wise compression achieves 1.25-1.63$\times$ better test accuracy. In contrast to CIFAR-10 results, layer-wise compression supersedes entire-model compression even at small ratios (i.e., $0.1\%$). This reaffirms that layer-wise compression remains more effective for models with a larger number of layers compared to previous experiments.

\subsection{Training Time}
We remark that our focus in this study is the convergence behavior of existing methods. Although total training time is an important factor, we do not present wall-clock results because (1) our unoptimized implementation does not represent a valid proof point for the level of training speedup that well-engineered compression methods can offer, and (2) the literature has already established that there are potentially large speedups to be achieved.
Indeed, our measurements show that performance depends on several factors including the model size and depth, the computational costs of compression and layer sizes. While layer-wise compression yields opportunities for overlapping communication and computation, we note that in some cases, overheads are better amortized by combining multiple invocations of the communication routines into a single one.
We leave it to future work to thoroughly study the performance implications of layer-wise vs. entire-model compression.

\section{Conclusion}
We identified a significant discrepancy between the theoretical analysis of the existing gradient compression methods and their practical implementation: while in practice compression is applied layer-wise, theoretical analysis is presumes compression is applied on the entire model.
We addressed the lack of understanding of converge guarantees in the layer-wise setting by proposing a bi-directional compression framework that encompasses both biased compression methods and unbiased ones as a special case. We proved tighter bounds on the noise (i.e., convergence) induced on SGD optimizers by layer-wise compression and showed that it is theoretically no worse than entire model compression. We implemented many common compression methods and evaluated their accuracy comparing layer-wise compression to entire-model compression. Conforming to our analysis, our results illustrated that in most cases, layer-wise compression performs no worse than entire-model compression, and in many cases it achieves better accuracy, in particular for larger models.

\bibliographystyle{aaai}
{\small
\bibliography{references}
}

\newpage
\onecolumn
\appendix 

\section{Convergence Analysis: Proofs}
\label{app:proofs}
In this section, we prove the claims in Section \ref{sec:conv}.

\begin{lemma*}
{\it Let $Q(\cdot):\R^d\to\R^d$ be layer-wise biased compression operator with $Q\eqdef (Q_1\;Q_2\cdots Q_L)$, such that, each $Q_j(\cdot):\R^{d_j}\to\R^{d_j}$ for $j=1,2,\cdots, L$ satisfies Assumption \ref{assum:compression} with $\Omega=\Omega_j$. Then we have 
\begin{eqnarray}\label{eq:q_sum2}
 \textstyle\mathbb{E}_Q\left(\|Q(x)\|_2^2\right)\le\sum_{1\le j\le L}(1+\Omega_j)\|x^j\|_2^2\le\max_{1\le j\le L}(1+\Omega_j)\|x\|_2^2.
\end{eqnarray}}
\end{lemma*}
\begin{proof}
Exploiting the layer-wise structure of $Q=(Q_1\;Q_2\cdots Q_L)$ we can write 
$$\mathbb{E}_Q\left(\|Q(x)\|_2^2\right)=\mathbb{E}_Q\left(\|(Q_1\;Q_2\cdots Q_L)(x)\|_2^2\right)\overset{{\rm By}\;\eqref{eq:kcontrac}}\le\sum_{1\le j\le L}(1+\Omega_j)\|x^j\|_2^2\le\max_{1\le j\le L}(1+\Omega_j)\|x\|_2^2.$$
Hence the result.
\end{proof}

\begin{lemma*}
{\it We note the following:
\begin{enumerate}
\item (For unbiased compression) If $\tilde{g}_k$ is unbiased (the case when $Q_M$ and $Q_W$ are unbiased), then
\begin{equation}
\mathbb{E} \left[\tilde{g}_k^\top \nabla f_k\right] = \mathbb{E} \|\nabla f_k\|_2^2. 
\end{equation}
Therefore, $\tilde{g}_k$ satisfies  Assumption~\ref{assum:decrease-second} with $\alpha=2$, $\|\cdot\|= \|\cdot\|_2$ and $ R_k = 0$.

\item  If $Q_M$ and $Q_W$ are both Random $k$ compression with compression factors $k_M$ and $k_W$, respectively, then
\begin{equation}
\mathbb{E} \left[\tilde{g}_k^\top \nabla f_k\right] = \frac{k_M k_W}{d^2} \mathbb{E} \|\nabla f_k\|_2^2.
\end{equation}
Therefore, $\tilde{g}_k$ satisfies  Assumption~\ref{assum:decrease-second} with $\alpha=2$, $\|\cdot\|=\frac{k_M k_W}{d^2} \|\cdot\|_2$, and $ R_k = 0$.

\item Let $Q_M$ be the layer-wise compression operator that is Random ${k_{M_j}}$ for each layer $j$. Similarly, $Q_W$ is also a layer-wise compression operator that is the same for all workers and for each layer $j$, it is Random ${k_{W_j}}$, then
\begin{equation}
\mathbb{E} \left[\tilde{g}_k^\top \nabla f_k\right] = \mathbb{E} \|\nabla f_k\|_{\mB}^2, 
\end{equation}
where $\textstyle{\mB= {\rm diag}\left(\frac{k_{M_1} k_{W_1}}{d_1^2}\mI_1\;\frac{k_{M_2} k_{W_2}}{d_2^2}\mI_2\cdots \frac{k_{M_L} k_{W_L}}{d_L^2}\mI_L \right)}$.
Therefore, $\tilde{g}_k$ satisfies  Assumption~\ref{assum:decrease-second} with $\alpha=2$, $\|\cdot\|= \|\cdot\|_{\mB}$ and $ R_k = 0$.
\item Let $Q_W$ be sign compression, then
\begin{equation}
\mathbb{E} \left[\tilde{g}_k^\top \nabla f_k\right] \ge \mathbb{E} \|\nabla f_k\|_1 +  R_k. 
\end{equation}
Therefore, $\tilde{g}_k$ satisfies Assumption~\ref{assum:decrease-second} with  $\alpha=1$, $\|\cdot\|= \|\cdot\|_{1}$, and $R_k = \cO\left(\frac{1}{BS}\right)$, where $BS$ is the size of used batch to compute the signSGD.
\end{enumerate}}
\end{lemma*}

\begin{proof}
\begin{enumerate}
\item If $\tilde{g}_k$ is unbiased, then $\mathbb{E}[\tilde{g}_k|x_k] = \nabla f_k$, whence 
\begin{equation}
\mathbb{E} \left[\tilde{g}_k^\top \nabla f_k | x_k \right] =  \|\nabla f_k\|_2^2.
\end{equation}
 Using tower property we conclude in this case that $\mathbb{E} \left[\tilde{g}_k^\top \nabla f_k \right] =  \mathbb{E} \|\nabla f_k\|_2^2,$
 therefore in this case  Assumption~\ref{assum:decrease-second} is satisfied with $\alpha=2$, $\|\cdot\|= \|\cdot\|_2$ and $ R_k = 0$.

\item If $Q_M$ and $Q_W$ are both Random $k$ with sparsification ratios $k_M$ and $k_W$, respectively:
We have $\tilde{g}_k=Q_M(\frac{1}{n}\sum_{i=1}^n Q_W(g_k^i))$. 
Therefore, by taking expectation based on the internal randomness of $Q_M$ we have 
\begin{eqnarray*}
\mathbb{E}_{Q_M}\left(\tilde{g}_k\right)&=&\mathbb{E}_{Q_M}\left( Q_M\left (\frac{1}{n}\sum_{i=1}^n Q_W(g_k^i)\right)\right)\\
&= & \frac{k_M}{d} \frac{1}{n}\sum_{i=1}^n Q_W(g_k^i). 
\end{eqnarray*}
Further taking expectation based on the internal randomness of the workers we have
\begin{eqnarray*}
\mathbb{E}_{Q_W} \mathbb{E}_{Q_M}\left(\tilde{g}_k\right) 
&= & \frac{k_M}{d} \frac{1}{n}\sum_{i=1}^n \mathbb{E}_{Q_W}(Q_W(g_k^i))\\
&=& \frac{k_M k_W}{d^2} \frac{1}{n}\sum_{i=1}^n g_k^i.
\end{eqnarray*}
 Therefore we conclude that 
\begin{equation}
\mathbb{E} \left[\tilde{g}_k^\top \nabla f_k\right] = \frac{k_M k_W}{d^2} \mathbb{E} \|\nabla f_k\|_2^2,
\end{equation}
 so in this case  Assumption~\ref{assum:decrease-second} is satisfied with $\alpha=2$, $\|\cdot\|=\frac{k_M k_W}{d^2} \|\cdot\|_2$, and $ R_k = 0$.

\item Let $Q_M$ be the layer-wise compression operator that is Random ${k_{M_j}}$ for each layer $j$. Similarly, $Q_W$ is also a layer-wise compression operator that is the same for all workers and for each layer $j$, it is Random ${k_{W_j}}$, then

We have $\tilde{g}_k=Q_M(\frac{1}{n}\sum_{i=1}^n Q_W(g_k^i))$. 
Therefore, by taking expectation based on the internal randomness of $Q_M$ we have 
\begin{eqnarray*}
\mathbb{E}_{Q_M}\left(\tilde{g}_k\right)&=& \sum_{j=1}^L\frac{k_{M_j}}{d_j}\left(\frac{1}{n}\sum_{i=1}^n Q_W(g_k^i))^j\right).
\end{eqnarray*}
Further taking expectation based on the internal randomness of the workers we have
\begin{eqnarray*}
\mathbb{E}_{Q_W} 
\mathbb{E}_{Q_M}\left(\tilde{g}_k\right)&=& \sum_{j=1}^L\frac{k_{M_j} k_{W_j}}{d_j^2}\left(\frac{1}{n}\sum_{i=1}^n (g_k^i)^j\right).
\end{eqnarray*}
 Therefore by taking the scalar product with the gradient and  taking the expectation  we conclude that 
\begin{equation*}
\mathbb{E} \left[\tilde{g}_k^\top \nabla f_k\right] =  \mathbb{E} \|\nabla f_k\|_{\mB}^2,
\end{equation*} 
where $\textstyle{\mB= {\rm diag}\left(\frac{k_{M_1} k_{W_1}}{d_1^2}\mI_1\;\frac{k_{M_2} k_{W_2}}{d_2^2}\mI_2\cdots \frac{k_{M_L} k_{W_L}}{d_L^2}\mI_L \right)}$.
Hence, in this case Assumption~\ref{assum:decrease-second} is satisfied with $\alpha=2$, $\|\cdot\|= \|\cdot\|_{\mB}$ and $ R_k = 0$.

\item This result is directly inspired from \cite{signsgd}. We quote it for completeness. If we consider signSGD as the compression method, then $\tilde{g}_k={\rm sign}(\frac{1}{n}\sum_i{\rm sign}({g}_k^i))$. We note that it is a deterministic and biased compression method. Therefore, as in \cite{signsgd}
\begin{eqnarray*}
\mathbb{E}\left[\tilde{g}_k^\top \nabla f_k| x_k \right]&\ge&\|\nabla f_k\|_{1}+\cO\left(\frac{1}{{BS}}\right),
\end{eqnarray*}
where $BS$ is the mini-batch size. Therefore, $R_k = \cO\left(\frac{1}{BS}\right)$, $\alpha=1$, and $\|\cdot\|=\|\cdot\|_{1}.$
\end{enumerate}
\end{proof}

Now we bound the compressed gradient $\tilde{g}_k.$
\begin{lemma*}
{\it Let Assumption \ref{assm:bdd_var} hold. 
With the notations defined above, we have
\begin{equation}\label{eq:secondorderterm2}
 \mathbb{E} \|\tilde{g}_k\|_2^2 \le \rho \|\nabla f_k\|_{\bf A}^2 + \sigma^2.
\end{equation}}
\end{lemma*}
\begin{proof}
Recall $\tilde{g}_k=Q_M(\frac{1}{n}\sum_{i=1}^n Q_W(g_k^i))$. 
Therefore, by using \eqref{eq:q_sum2} and taking expectation based on the internal randomness of $Q_M$ we have 
\begin{eqnarray*}
\mathbb{E}_{Q_M}\left(\|\tilde{g}_k\|_2^2\right)&=&\mathbb{E}_{Q_M}\left(\left \|Q_M\left (\frac{1}{n}\sum_{i=1}^n Q_W(g_k^i)\right)\right \|_2^2\right)\\
&\le & \sum_{j=1}^L(1+\Omega_M^j)\left\|\frac{1}{n}\sum_{i=1}^n (Q_W(g_k^i))^j\right \|_2^2\\
&\overset{\|\sum_{i=1}^na_i\|_2^2\le n\sum_{i=1}^n\|a_i\|_2^2}\le & \frac{1}{n}\sum_{i=1}^n \sum_{j=1}^L(1+\Omega_M^j)\left\|(Q_W(g_k^i))^j\right \|_2^2,
\end{eqnarray*}
 
Further taking expectation based on the internal randomness of the workers we have
\begin{eqnarray*}
\mathbb{E}_{Q_W}\left(\mathbb{E}_{Q_M}\left(\|\tilde{g}_k\|_2^2\right)\right)&\le &\frac{1}{n}\sum_{i=1}^n \sum_{j=1}^L(1+\Omega_M^j)\mathbb{E}_{Q_W}\left(\|(Q_W(g_k^i))^j\|_2^2\right)\\
&\le&\frac{1}{n}\sum_{i=1}^n \sum_{j=1}^L(1+\Omega_M^j)(1+\Omega_W^j)\|(g_k^i)^j\|_2^2\\
&=&\frac{1}{n}\sum_{i=1}^n\|g_k^i\|_{\mA}^2.
\end{eqnarray*}
Now we take the final expectation and by using tower property of the expectation we have
\begin{eqnarray*}
\mathbb{E}\left(\|\tilde{g}_k\|_2^2\right)&\le & \frac{1}{n}\sum_{i=1}^n\mathbb{E}(\|g_k^i\|_{\mA}^2)\\
&\overset{{\rm By\;Assumption\;} \ref{assm:bdd_var}}\le &\rho \|\nabla f_k\|_{\bf A}^2 + \sigma^2.
\end{eqnarray*}

Hence the result. 

\end{proof}

\begin{proposition*}
{\it With the notations and the framework defined before, the iterates of \eqref{iter:sgd} satisfy 
\begin{eqnarray}\label{eq:main-result2}
 \eta_k\left(\mathbb{E}\|\nabla f_k\|^\alpha - \frac{\cL\eta_k}{2}\mathbb{E}\|\nabla f_k\|_{\bf A}^2 \right)&\leq  \mathbb{E}(f_k-f_{k+1})-\eta_k R_k +\tfrac{\cL\eta_k^2\sigma^2}{2}.
\end{eqnarray}}
\end{proposition*}
\begin{proof}
From the $\cL$-smoothness of the function $f$ we have  
\begin{eqnarray*}
f_{k+1}&\le & f_{k}+\nabla f_k^\top (x_k-x_{k+1}) + \frac{\eta_k^2\cL }{2}\|x_k-x_{k+1}\|^2\\
&\overset{{\rm By\;}\eqref{iter:sgd}}=& f_{k}-\eta_k\nabla f_k^\top \tilde{g}_k + \frac{\eta_k^2\cL }{2}\|\tilde{g}_k\|^2,
\end{eqnarray*} 
which after taking the expectation and by using Lemma \ref{eq:secondorderterm2} and Assumption \ref{assum:decrease-second}, reduces to
\begin{eqnarray*}
\mathbb{E}(f_{k+1}) \leq \mathbb{E}(f_k) -\eta_k\mathbb{E} \|\nabla f_k\|^\alpha - \eta_kR_k  + \frac{\cL \eta_k^2}{2}(\mathbb{E}\|\nabla f_k\|_{\bf A}^2+\sigma^2).
\end{eqnarray*}
Rearranging the terms we get the final result.
\end{proof}

\begin{proposition*}
(Special case.)
{\it Consider $\alpha=2$, and $ R_k = 0$. Let $C>0$ be the constant due to the equivalence between the norms  $\|\cdot\|$ and $\|\cdot\|_{\mA}$, and $K>0$ be the number of iterations. If $\eta_k=\eta = \cO\left(\tfrac{1}{\sqrt{K}}\right)<\tfrac{2}{LC \rho}$ then 
\begin{eqnarray*}
\tfrac{\sum_{k=1}^K \mathbb{E}\|\nabla f_k\|^2}{K}\leq \cO\left(\tfrac{1}{\sqrt{K}}\right). 
\end{eqnarray*}}
\end{proposition*}
\begin{proof}
Set $\eta_k=\eta.$ From (\ref{eq:main-result2}) we have
\begin{equation*}
\eta \left(\mathbb{E}\|\nabla f_k\|^2 -  \frac{L \rho \eta}{2} \mathbb{E}\|\nabla f_k\|_{\bf A}^2 \right)\leq \mathbb{E}(f_k) -\mathbb{E}(f_{k+1})   +\frac{\cL \sigma^2 \eta^2}{2}.
\end{equation*}
 From the equivalence between $\|\cdot\|$ and $\|\cdot\|_{\mA}$ norms, there exist a constant $C>0$ such that $\mathbb{E}\|\nabla f_k\|_{\mA}^2\le C\mathbb{E}\|\nabla f_k\|^2$ and we have 
 \begin{equation*}
 \mathbb{E} \|\nabla f_k\|^2\left(1 -  \frac{\cL C \rho \eta}{2}\right) \le 
    \mathbb{E} \|\nabla f_k\|^2 -  \frac{\cL \rho \eta}{2} \mathbb{E}\|\nabla f_k\|_{\bf A}^2.
  \end{equation*}
Let $a \eqdef \left(1 -  \frac{L C \rho \eta}{2}\right) > 0$, then by combining the previous two inequalities gives 
\begin{eqnarray*}
\mathbb{E} \|\nabla f_k\|^2 \leq \frac{\mathbb{E}(f_k) - \mathbb{E}(f_{k+1})}{\eta a}  +\frac{\cL \sigma^2 \eta}{2 a}. 
\end{eqnarray*}
 By unrolling the recurrence on this inequality we get
 \begin{eqnarray*}
\frac{\sum_{k=1}^K \mathbb{E} \|\nabla f_k\|^2}{K} \leq \frac{f_0 - \mathbb{E}(f_{K+1})}{ K \eta a}  +\frac{\cL \sigma^2 \eta}{2 a}. 
\end{eqnarray*}
 Since $\eta = \cO\left(\frac{1}{\sqrt{K}}\right)$ and $\mathbb{E}(f_{k+1}) \ge f^\star$  then
  \begin{eqnarray*}
\frac{\sum_{k=1}^K \mathbb{E} \|\nabla f_k\|^2}{K} \leq \cO\left(\frac{1}{\sqrt{K}}\right). 
\end{eqnarray*}
\end{proof}

\begin{proposition*}
(General case.) 
{\it Assume $\|\nabla f_k\|\le G$. Let $C>0$ be the constant coming from the equivalence between  $\|\cdot\|$ and $\|\cdot\|_{\mA}$. Let $\eta_k=\eta=\cO\left(\frac{1}{\sqrt{K}}\right)<\frac{2}{\cL C\rho G^{2-\alpha}}$. Let $a \eqdef \left(1 -  \frac{\cL C \rho G^{2-\alpha} \eta}{2}\right) > 0$, then
  \begin{eqnarray*}
\frac{\sum_{k=1}^K \mathbb{E} \|\nabla f_k\|^\alpha}{K} \leq \cO\left(\frac{1}{\sqrt{K}}\right) + \frac{\sum_{k=1}^K  R_k}{a K}.
\end{eqnarray*}}
\end{proposition*}

\begin{proof}
We assume $\|\nabla f_k\|\le G.$ From the equivalence between  $\|\cdot\|$ and $\|\cdot\|_A$ there exist a constant $C>0$ such that \eqref{eq:main-result2} reduces to
\begin{equation*}
 \mathbb{E} \|\nabla f_k\|^\alpha \left(1 -  \frac{\cL C \rho G^{2-\alpha} \eta}{2}\right) \le 
    \mathbb{E} \|\nabla f_k\|^\alpha -  \frac{\cL C \rho \eta}{2} \mathbb{E}\|\nabla f_k\|_{\bf A}^2.
  \end{equation*}
 By injecting the previous inequality in (\ref{eq:main-result2}) we get
\begin{eqnarray*}
\mathbb{E} \|\nabla f_k\|^\alpha \leq \frac{\mathbb{E}(f_k) - \mathbb{E}(f_{k+1})}{\eta a}  + \frac{\cL \sigma^2 \eta}{2 a} + \frac{  R_k}{a}. 
\end{eqnarray*}
 By unrolling the recurrence on this inequality we get
 \begin{eqnarray*}
\frac{\sum_{k=1}^K \mathbb{E} \|\nabla f_k\|^\alpha}{K} \leq \frac{f_0 - \mathbb{E}(f_{k+1})}{ K \eta a}  +\frac{\cL \sigma^2 \eta}{2 a} + \frac{\sum_{k=1}^K  R_k}{a K}. 
\end{eqnarray*}
 Since $\eta = \cO\left(\frac{1}{\sqrt{K}}\right)$ and $\mathbb{E}(f_{k+1}) \ge f^\star$  then 
  \begin{eqnarray*}
\frac{\sum_{k=1}^K \mathbb{E} \|\nabla f_k\|^\alpha}{K} \leq \cO\left(\frac{1}{\sqrt{K}}\right) + \frac{\sum_{k=1}^K  R_k}{a K}.
\end{eqnarray*}
\end{proof}

\end{document}